\documentclass[11pt,letterpaper]{amsart}
\usepackage[leqno]{amsmath}
\usepackage{amsfonts,amsthm,amssymb}
\usepackage{mathrsfs}
\usepackage{epsfig}
\usepackage{graphicx}
\usepackage{subfigure}
\usepackage{url}
\usepackage{hyperref}
\usepackage[left=1.08in,right=1.08in,top=1.08in,bottom=1.08in,foot=0.6in]{geometry}
\usepackage{cite}
\usepackage{xspace}
\usepackage{enumerate}

\newtheorem{lemma}{Lemma}
\newtheorem{proposition}{Proposition}
\newtheorem{theorem}{Theorem}
\newtheorem{corollary}{Corollary}

\theoremstyle{definition}
\newtheorem{definition}{Definition}

\newtheorem{remark}{Remark}

\newtheorem{conjecture}{Conjecture}
\newtheorem{example}{Example}

\newcommand{\mcs}{\mathcal S}	 		
\newcommand{\mcg}{\mathcal G}		   
\newcommand{\mcp}{\mathcal P}			

\newcommand{\product}{\square}

\DeclareMathOperator{\dens}{dens}	      
\DeclareMathOperator{\vcd}{VC-dim}		  
\DeclareMathOperator{\arbo}{a}			  
\DeclareMathOperator{\vcdens}{VC-dens}
\DeclareMathOperator{\mad}{mad}

\DeclareMathOperator{\dd}{dd}	
\newcommand{\DD}{\dd^*} 
\newcommand{\cm}[1]{\omega(#1)} 

\newcommand{\lenum}[1]{{(#1)}}	

\newcommand{\tip}{leaf\xspace}
\newcommand{\tips}{leaves\xspace}

\newcommand{\ceil}[1]{\left\lceil{#1}\right\rceil}

\begin{document}

\centerline{\Large\bf On density of subgraphs of Cartesian products}

\vspace{10mm}
\centerline{Victor Chepoi, Arnaud Labourel, and S\'ebastien Ratel}

\medskip
\begin{small}
	\medskip
	\centerline{Aix Marseille Univ, Universit\'e de Toulon, CNRS, LIS, 
	Marseille, France}

	\centerline{\texttt{\{victor.chepoi, arnaud.labourel, 
	sebastien.ratel\}@lis-lab.fr}}
\end{small}

\bigskip\bigskip\noindent
{\footnotesize {\bf Abstract.}
	In this paper, we extend two classical results about the density of 
	subgraphs of hypercubes to subgraphs $G$ of Cartesian products 
	$G_1\product\cdots\product G_m$ of arbitrary connected graphs. Namely, we 
	show that $\frac{|E(G)|}{|V(G)|} \le \ceil{2\max\{\dens(G_1), \ldots, 
	\dens(G_m)\}} \log|V(G)|$, where $\dens(H)$ is the maximum ratio 
	$\frac{|E(H')|}{|V(H')|}$ taken over all subgraphs $H'$ of $H$.  We 
	introduce the notions of VC-dimension $\vcd(G)$ and VC-density $\vcdens(G)$ 
	of a subgraph $G$ of a Cartesian product $G_1\product\cdots\product G_m$, 
	generalizing the classical Vapnik-Chervonenkis dimension of set-families 
	(viewed as subgraphs of hypercubes). We prove that if $G_1,\ldots,G_m$ 
	belong to the class ${\mathcal G}(H)$ of all finite connected graphs not
	containing a given graph $H$ as a minor, then  for any subgraph $G$ of
	$G_1\product\cdots\product G_m$	the sharper inequality 
	$\frac{|E(G)|}{|V(G)|} \le \mu(H) \cdot \vcd^*(G)$ holds, where $\mu(H)$ is 
	the 
	supremum of the densities of the graphs from ${\mathcal G}(H)$. We refine 
	and sharpen these two results to several specific graph classes. We also
	derive  upper bounds (some of them polylogarithmic) for the size of 
	adjacency labeling schemes of subgraphs of Cartesian products.
}

\section{Introduction} \label{introduction}
A folklore result (see for example \cite[Lemma 3.2]{HaImKl} and \cite{Gr}) 
asserts that if $G=(V,E)$ is
an induced $n$-vertex subgraph of the $m$-dimensional hypercube $Q_m$, then 
$\frac{|E|}{|V|}\le
\log n$. This inequality together with the fact that
this class of graphs is closed by taking induced subgraphs immediately implies 
that $n$-vertex
subgraphs of hypercubes have $O(\log n)$ density, degeneracy, arbority, and 
consequently admit
$O(\log^2 n)$ adjacency labeling schemes. This {\it density bound}  
$\frac{|E|}{|V|}\le \log n$ has
been refined and sharpened in several directions and these improvements lead to 
important
applications and results.

On the one hand, the {\it  edge-isoperimetric problem}	for hypercubes 
\cite{Bezrukov,Har} asks for
any integer $1\le n\le 2^m$ to find an $n$-vertex subgraph $G$ of the 
$m$-dimensional cube
$Q_m$ with the smallest edge-boundary $\partial G$, i.e., with the minimum  
number of edges of
$Q_m$ running between $G$ and its complement in $Q_m$. Since the hypercubes 
are regular
graphs, minimizing the boundary $\partial G$ of $G$  is equivalent to 
maximizing the number of
edges of $G$, and thus to maximizing the density $\frac{|E|}{|V|}$ of $G$. The 
classical result by
Harper \cite{Har}
nicely characterizes the  solutions of the edge-isoperimetric problem for 
hypercubes: for any $n$,
this is the subgraph of $Q_m$ induced by the initial segment of length $n$ of 
the {\it lexicographic
	numbering} of the vertices of $Q_m$.  One elegant way of proving this 
	result is using the
operation
of compression  \cite{Har_book}.  For a generalization of these results and 
techniques to subgraphs
of Cartesian products of other regular graphs and for applications of 
edge-isoperimetric problems,
see the book by Harper \cite{Har_book} and the survey by Bezrukov 
\cite{Bezrukov}.

On the other hand, to any set family (concept class) $\mathcal S\subseteq \{ 
0,1\}^m$ with $n$ sets
one can associate the subgraph $G({\mathcal S})=(V,E)$ of the hypercube $Q_m$ 
induced by the
vertices corresponding to the sets of $\mathcal S$, i.e., $V=\mathcal S$. This 
graph $G({\mathcal
	S})$ is called the {\it 1-inclusion graph} of $\mathcal S$; 1-inclusion 
	graphs have numerous
applications in computational learning theory, for example, in prediction 
strategies \cite{HaLiWa} and
in sample compression schemes \cite{KuWa}. Haussler, Littlestone, and Warmuth 
\cite{HaLiWa}
proposed a prediction strategy for concept classes based on their 1-inclusion 
graph (called the {\it
	1-inclusion prediction strategy}) as a natural approach to the prediction 
	model of learning. They
provided an upper bound on the worst-case expected risk of the 1-inclusion 
strategy for a concept
class $\mathcal S$ by the density of its 1-inclusion graph $G(\mathcal S)$ 
divided by $n$. Moreover,
\cite[Lemma 2.4]{HaLiWa} establishes  a sharp upper bound $\frac{|E|}{|V|}\le 
\vcd({\mathcal S})$ on
the density of $G(\mathcal S)$, where $\vcd(\mathcal S)$ is the 
Vapnik-Chervonenkis dimension of
${\mathcal S}$.  While $\vcd(\mathcal S)\le \log n$ always holds, for some
concept classes
$\vcd(\mathcal S)$ is much smaller than $\log n$ and thus the inequality 
$\frac{|E|}{|V|}\le
\vcd({\mathcal S})$ presents a significant improvement over the folklore 
inequality $\frac{|E|}{|V|}\le
\log n$. This is the case of maximum concept classes \cite{KuWa} and, more 
generally, of lopsided
systems \cite{BaChDrKo,BoRa,La}. In this case,  $\vcd(\mathcal S)$ is exactly 
the dimension of the
largest subcube of $G(\mathcal S)$ and this dimension may not depend at all on 
the number of sets
of $\mathcal S$. Haussler \cite{Hau} presented an elegant proof of the 
inequality $\frac{|E|}{|V|}\le
\vcd({\mathcal S})$ using the shifting (push-down) operation (which can be 
compared with the
compression operation  for edge-isoperimetric problem). He used this density 
result to give an
upper bound on the $\epsilon$-packing number (the maximum number of disjoint 
balls of a radius
$\epsilon$ of  $Q_m$ with centers at the vertices of $G({\mathcal S})$), and 
this result can be
viewed  as a far-reaching generalization of the classical Sauer lemma
\cite{Sauer}.   The inequality of
Haussler et al. \cite{HaLiWa} as well as the classical notions of
VC-dimension and the Sauer lemma have
been subsequently extended to the subgraphs of Hamming graphs, i.e., from 
binary alphabets to
arbitrary alphabets; see \cite{HaLo,Pollard_book,Natarajan}. Cesa-Bianchi and 
Haussler \cite{CBHa}
presented a graph-theoretical generalization of the Sauer Lemma for the
$m$-fold $F^m=F\product
\cdots\product F$ Cartesian products of arbitrary undirected graphs $F$.

In this paper, we extend the inequalities $\frac{|E|}{|V|}\le \log n$ and 
$\frac{|E|}{|V|}\le \vcd(\mathcal
S)$ to $n$-vertex subgraphs $G=(V,E)$ of Cartesian products
$\Gamma:=G_1\product\cdots \product
G_m$ of arbitrary connected graphs $G_1,\ldots,G_m$. Namely, in Theorem
\ref{thm_subgraphs_products1}, we show that the density of $G$ 
is at most $2\log n$ times the largest density of a factor of $\Gamma$. To 
extend the density result
of \cite{HaLiWa}, we define the notions of (minor and induced) VC-dimensions 
and VC-densities of
subgraphs $G$ of arbitrary Cartesian products and show that if all factors 
$G_1,\ldots, G_m$ do not
contain a fixed subgraph $H$ as a minor, then the density of any subgraph 
$G=(V,E)$ of $\Gamma$
is at most $\mu(H)$ times the VC-dimension $\vcd^*(G)$ of $G$: 
$\frac{|E|}{|V|}\le \mu(H) \cdot \vcd^*(G)$,
where $\mu(H)$ is a constant such that any graph not containing $H$ as a minor 
has density at
most $\mu(H)$ (it is well known \cite{Die} that if $r:=|V(H)|$, then $\mu(H)\le 
cr\sqrt{\log r}$ for a
universal constant $c$). We conjecture that in fact $\frac{|E|}{|V|}\le 
\vcdens^*(G)$ holds, where
$\vcdens^*(G)$ is the VC-density of $G$. We consider several classes of graphs 
for which sharper
inequalities hold. Since by Nash-Williams's theorem \cite{NaWi} all such 
inequalities provide upper
bounds for arboricity and since by a result of Kannan, Naor, and Rudich 
\cite{KaNaRu} bounded
arboricity implies bounded adjacency labeling schemes, in the last section of 
the paper we present
the applications of our results to the design of compact adjacency labeling 
schemes for subgraphs
of Cartesian products (which was one of our initial motivations).

The {\it canonical metric representation} theorem of Graham and Winkler 
\cite{GrWi} asserts that
any  connected finite graph $G$ has a unique isometric embedding into the 
Cartesian product
$\Pi_{i=1}^m G_{i}$ in which each factor $G_{i}$ is prime (i.e., not further 
decomposable this way)
and this representation can be computed efficiently; for proofs and algorithms, 
see the books
\cite{DeLa,HaImKl}.  Thus our results have a general nature and show that it 
suffices to bound the
density of prime graphs. For many classes of graphs occurring in metric graph 
theory
\cite{BaCh_survey}, the prime graphs have special structure. 
For example, the primes for isometric subgraphs of hypercubes (which have been 
characterized in a
nice way by Djokovi\'{c} \cite{Djokovic}) are the $K_2$. Thus the density of 
isometric subgraphs of
hypercubes (and more generally, of subgraphs of hypercubes) is upper bounded by 
their
VC-dimension. Shpectorov \cite{Shp} proved that the primes  of graphs which 
admit a scale
embedding into a hypercube are exactly the subgraphs of octahedra and isometric 
subgraphs of
halved cubes. In Section \ref{sect_consequences}, we will show how to bound the 
density of
subgraphs of Cartesian products of octahedra. In another paper 
\cite{ChLaRa_halfcubes} we will
define an appropriate notion of VC-dimension for subgraphs of halved cubes and 
we will use it to
upper bound the density of such graphs. The papers \cite{BrChChKoLaVa} and 
\cite{BrChChGoOs}
investigate  the local-to-global structure of graphs which are retracts of 
Cartesian products of
chordal graphs, bridged  and weakly bridged graphs, respectively (bridged 
graphs are the graphs in
which all isometric cycles have length 3). 2-Connected chordal  graphs and 
bridged or weakly
bridged graphs are prime. Notice also that the bridged and weakly bridged 
graphs are dismantlable
(see, for example, \cite[Section 7]{BrChChGoOs}). In Section 
\ref{sect_consequences}, we present
sharper density inequalities for subgraphs of Cartesian products of chordal 
graphs and of
dismantlable graphs, which can be directly applied to the classes of graphs from
\cite{BrChChKoLaVa} and \cite{BrChChGoOs}. For other such classes of graphs 
occurring in metric
graph theory, see the survey \cite{BaCh_survey} and the papers 
\cite{ChChHiOs,Cha}.

\section{Preliminaries}  \label{preliminaries} 
In this section, we define the basic notions and concepts used throughout the 
paper. Some specific
notions
(some classes of graphs, adjacency labeling schemes, etc) will be introduced 
when appropriate.

\subsection{Basic definitions} 

\subsubsection{Density}
All graphs $G=(V,E)$ occurring in this note are finite, undirected, and simple. 
The {\it closed neighbourhood} of a vertex $v$ is denoted by $N[v]$ and 
consists of $v$ and the
vertices adjacent to $v$. The degree $d(v)$ of a vertex $v$ is the number of 
edges of $G$ incident
to $v$. The number $d(G)=\frac{1}{|V|}\sum_{v\in G} d(v)=\frac{2|E|}{|V|}$ is 
the \emph{average
	degree} of $G$.

The \emph{maximum average degree} $\mad(G)$ of $G$ is the maximum average 
degree of a
subgraph $G'$ of $G$: $\mad(G)=\max \{ d(G'): G' \text{ is a subgraph of } G\}$.

The \emph{density $\dens(G)$} of $G$ will be the maximal ratio 
$|E(G')|/|V(G')|$ over all its
subgraphs $G'$.

Density and maximum average degree are closely related, namely $\dens(G) = 
\frac{\mad(G)}{2} =
\max\left\{\frac{d(G')}{2} : G'\subseteq G\right\}$ holds, but they quantify
different aspects of $G$. We will use both numbers, depending on the 
circumstances: $\dens(G)$
will be used to express a global parameter of $G$ in a result, and $\mad(G)$
will be used in proofs when we have to look at a local parameter (degrees) of 
$G$. We will use the
following simple observation:

\begin{lemma}\label{prop:2vertexmad}
	Let $G$ be a simple and connected graph. Then $G$ has two vertices of 
	degree at most
	$\ceil{\mad(G)}$.
\end{lemma}
\begin{proof} Let $\mu(G):=\ceil{\mad(G)}$ and $n$ the number of vertices of 
$G$. Assume for
	contradiction that there exists connected graph $G$ with one
	vertex $v_0$ of  degree at most $\mu(G)$ and $n-1$ vertices of  degree at 
	least $\mu(G)+1$. Since $G$ is connected,
$d(v_0)\geq 1$ and we obtain
	$\mad(G)\geq 
	d(G)\geq\frac{(\mu(G)+1)(n-1)+1}{n}=\mu(G)+1-\frac{\mu(G)}{n}>\mu(G)$, 
	leading to a
	contradiction.
\end{proof}

\subsubsection{Cartesian products \cite{DeLa,HaImKl}}
Let $G_{1},\ldots,G_{m}$ be a family of $m$ connected graphs. The 
\emph{Cartesian product}
$\Gamma:=\prod_{i=1}^m G_i=G_1\product\cdots\product G_m$ is a graph
defined on the set of all
$m$-tuples $(x_1,\ldots,x_m)$, $x_{i} \in V(G_{i})$, where two vertices 
$x=(x_1,\ldots,x_m)$ and
$y=(y_1,\ldots,y_m)$ are adjacent if and only if there exists an index $1\le 
j\le m$ such that $x_{j}
y_{j} \in E(G_{j})$ and $x_{i} = y_{i}$ for all $i\neq j$. If $uv$ is an edge 
of the factor $G_i$, then all
edges of $\Gamma$ running between two vertices of the form
$(v_1,\ldots,v_{i-1},u,v_{i+1},\ldots,v_m)$ and 
$(v_1,\ldots,v_{i-1},v,v_{i+1},\ldots,v_m)$  will be called
\emph{edges of type $uv$}. A factor $G_i$ is called a {\it non-trivial factor} 
of $\Gamma$ if $G_i$ contains at least two vertices.
A graph $G$ is \emph{prime} if it can not be represented as a Cartesian product 
of two non-trivial graphs.

The \emph{$m$-dimensional hypercube} $Q_m$ is the Cartesian product of $m$ 
copies of $K_2$ with $V(K_2)=\{ 0,1\}$, i.e., $Q_m=K_2\product \cdots \product 
K_2$.
Equivalently, $Q_m$ has the subsets $S$ of a set $X$ of size $m$ as the 
vertex-set and two such sets $A$ and $B$ are adjacent in $Q_m$ if and only if 
$|A\Delta B|=1$.

A \emph{subproduct} $\Gamma'$ of a Cartesian  product $\Gamma=\prod_{i=1}^m 
G_i$ is a product
such that $\Gamma'=\prod_{j=1}^k G'_{i_j}$, where each $G'_{i_j}$ is a 
connected non-trivial
subgraph
of
$G_{i_j}$.  A subproduct $\Gamma'=\prod_{j=1}^k G'_{i_j}$ in which each factor 
$G'_{i_j}$ is an edge
of $G_{i_j}$  is called a {\it cube-subproduct} of $\Gamma$.

Given a vertex $v'=(v'_{i_1},\ldots,v'_{i_k})$ of $\Gamma'$, we  say that a 
vertex $v=(v_1,\ldots,v_m)$
of
$\Gamma$  is an \emph{extension} of $v'$ if $v_{i_j}=v'_{i_j}$ for all 
$j\in\{1,\ldots,k\}$. 
We denote by $F(v')$ the set of all extensions
$v$ in $\Gamma$ of a vertex $v'$ of a subproduct $\prod_{j=1}^k G'_{i_j}$ and 
call $F(v')$ the
\emph{fiber} of $v'$ in the product $\Gamma$ (see Figure 
\ref{fig_sect_cartesian_product} for an
illustration of these notions).

Let $G$ be a subgraph of a Cartesian product $\Gamma$ and $\Gamma'$ be a 
subproduct of
$\Gamma$. The {\it trace} of $V(G)$ on  $V(\Gamma')$ consists of all vertices 
$v'$ of $\Gamma'$
such that $F(v')\cap V(G)\ne \emptyset$. The \emph{projection} of $G$ on 
$\Gamma'$ is  the
subgraph $\pi_{\Gamma'}(G)$ of $\Gamma'$ induced by the trace of $V(G)$ on 
$V(\Gamma')$. We
will denote the projection of $G$ on the $i$th factor $G_i$ of $\Gamma$ by 
$\pi_i(G)$ instead of
$\pi_{G_i}(G)$.

The following lemma must be well-known, but we have not found it in the 
literature. Its proof was
communicated to us by Fran\c{c}ois Dross (which we would like to acknowledge).

\begin{figure}
	\centering
	\begin{minipage}{0.49\textwidth}
		\centering
		\includegraphics[width=0.9\textwidth]{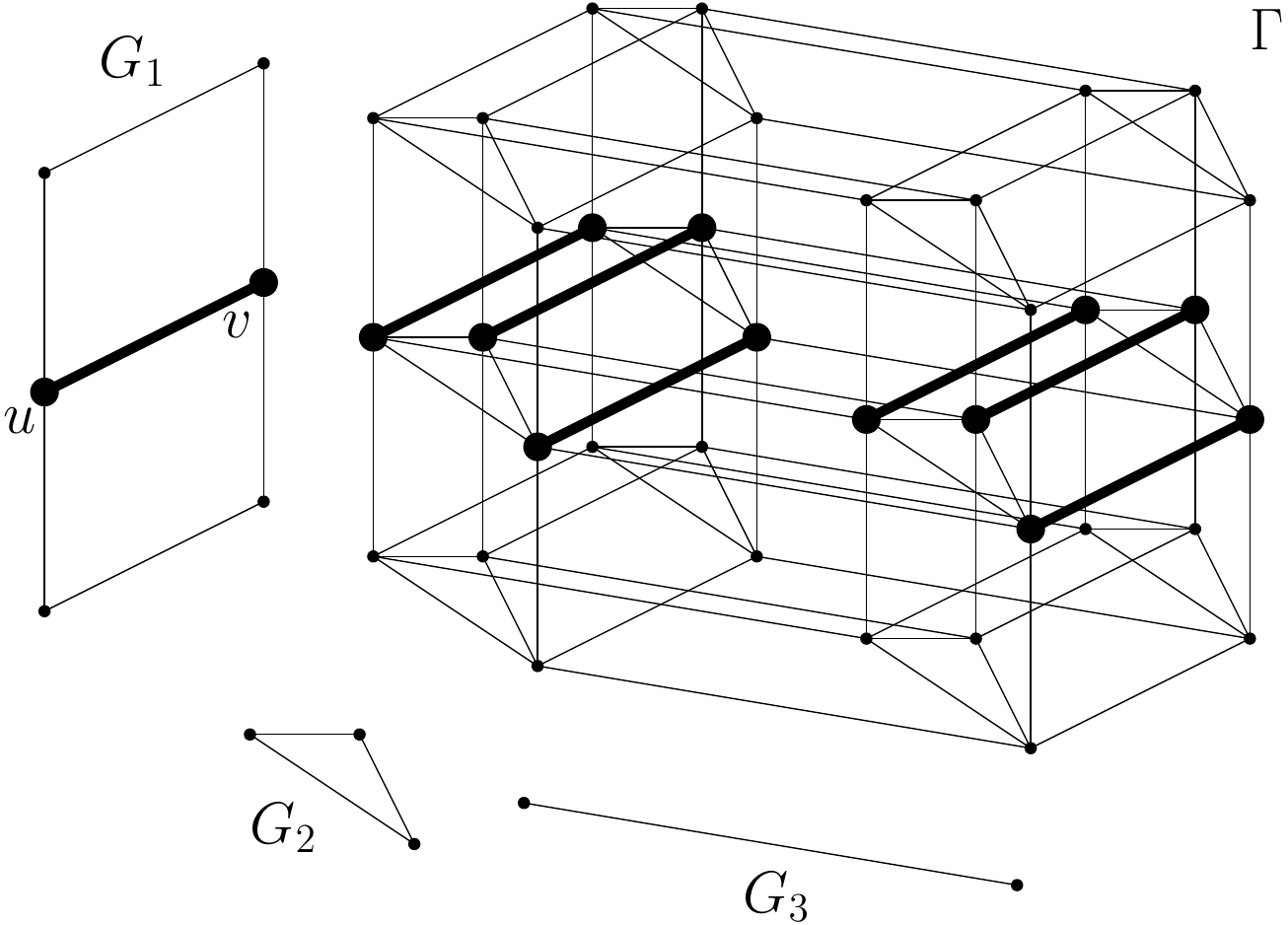}
	\end{minipage}
	\begin{minipage}{0.49\textwidth}
		\centering
		\includegraphics[width=0.9\textwidth]{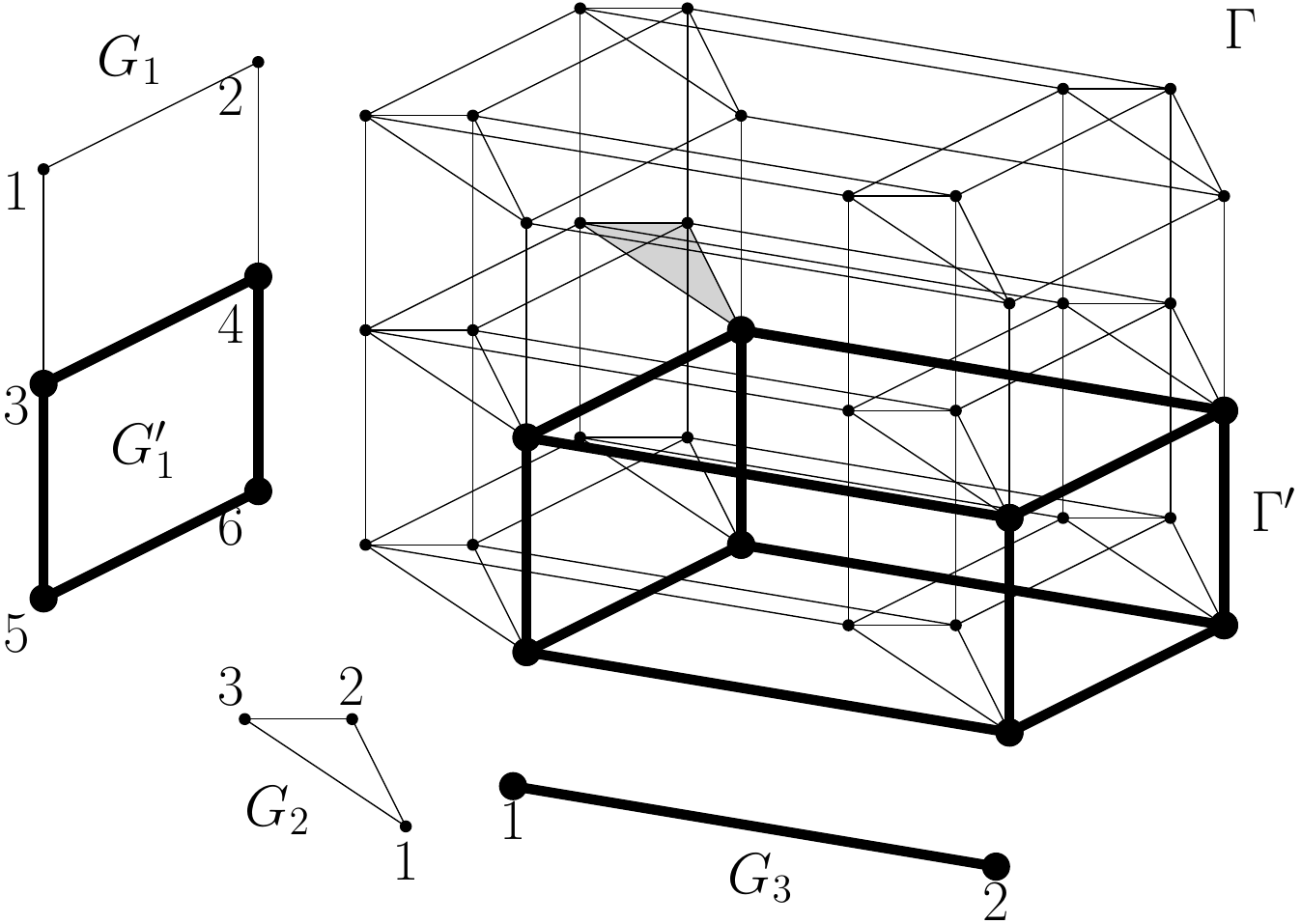}
	\end{minipage}
	\caption{\label{fig_sect_cartesian_product} On left, an example of a 
	Cartesian product
		$\Gamma := G_1 \product G_2 \product G_3$ and an edge $uv \in
		E(G_1)$ and its copies in
		$\Gamma$
		(the edges of type $uv$ of $\Gamma$). On right, a subproduct (which is 
		also a
		cube-subproduct) $\Gamma' := G_1' \product G_3'$ of $\Gamma$ and
		the fiber (in gray) $F((4,1))
		:= \{(4,1,1),(4,2,1),(4,3,1)\}$ of the vertex $v = (4,1)$ of $\Gamma'$.}
\end{figure}

\begin{lemma} \label{lem_product_density} If $\Gamma=\prod_{i=1}^m G_i$, then
	$\dens(\Gamma)=\sum_{i=1}^m \dens(G_i)$, i.e., if $G'_i=(V'_i,E'_i)$ is a 
	densest subgraph of
	$G_i=(V_i,E_i)$, then $\Gamma'=\prod_{i=1}^m G'_i$ is a densest subgraph of 
	$\Gamma$.
\end{lemma}

\begin{proof}
	The graph $\Gamma'$ has $|V'_1|\cdot\ldots\cdot|V'_m|$ vertices and 
	$\sum_{i=1}^m
	|E_i'|\cdot(\prod_{j=1, j\ne i}^m |V'_j|)$ edges. Therefore
	$\frac{|E'|}{|V'|}=\sum_{i=1}^m\frac{|E'_i|}{|V'_i|}=\sum_{i=1}^m\dens(G_i),$
	 showing that
	$\dens(\Gamma)\ge \sum_{i=1}^m\dens(G_i)$.

	We will prove now the converse inequality $\dens(\Gamma)\le 
	\sum_{i=1}^m\dens(G_i)$. Since the
	Cartesian product operation is associative, it suffices to prove this 
	inequality for two factors.
	Let $\Gamma=G_1\product G_2$ and let $G=(V,E)$ be a densest subgraph of
	$\Gamma$. We will call
	the edges of $G$ arising from $G_1$ horizontal edges and those arising from 
	$G_2$ vertical edges
	and denote the two edge-sets by $E_h$ and $E_v$. Then
	$\frac{|E|}{|V|}=\frac{|E_h|}{|V|}+\frac{|E_v|}{|V|}$. Thus it suffices to 
	prove the inequalities
	$\frac{|E_h|}{|V|}\le \dens(G_1)$ and $\frac{|E_v|}{|V|}\le \dens(G_2)$. 
	First, we will establish
	the first inequality. For each vertex $x\in V_2$ we will denote by 
	$L(x)=(V(x),E(x))$ the subgraph
	of $G$ induced by the vertices  of $G$ having $x$ as their second 
	coordinate (this
	subgraph is called the {\it $x$-layer} of $G$, see Fig. \ref{fig_layers}(a) 
	for an illustration). All
	horizontal edges of $G$ and all vertices of $G$
	are distributed in such layers and the layers are pairwise disjoint. Each 
	layer $L(x)$ is isomorphic 
    to a subgraph of
	$G_1$, thus $\dens(L(x))\le \dens(G_1)$ for any $x\in V_2$. Therefore,
	$$\frac{|E_h|}{|V|}=\frac{\sum_{x\in V_2} |E(x)|}{\sum_{x\in V_2} 
	|V(x)|}\le \max_{x\in V_2}
	\left\{\frac{|E(x)|}{|V(x)|}\right\}\le \dens(G_1),$$
	as required (we used the inequality $\frac{\sum_{i=1}^m a_i}{\sum_{i=1}^m 
	b_i}\le
	\max\left\{\frac{a_i}{b_i}\right\}$ for nonnegative numbers 
	$a_1,\ldots,a_m$ and $b_1,\ldots,b_m$).
	The second inequality can be shown with a similar argument.
\end{proof}

\begin{figure}
	\begin{minipage}{0.49\textwidth}
		\centering
		\includegraphics[width=0.5\textwidth]{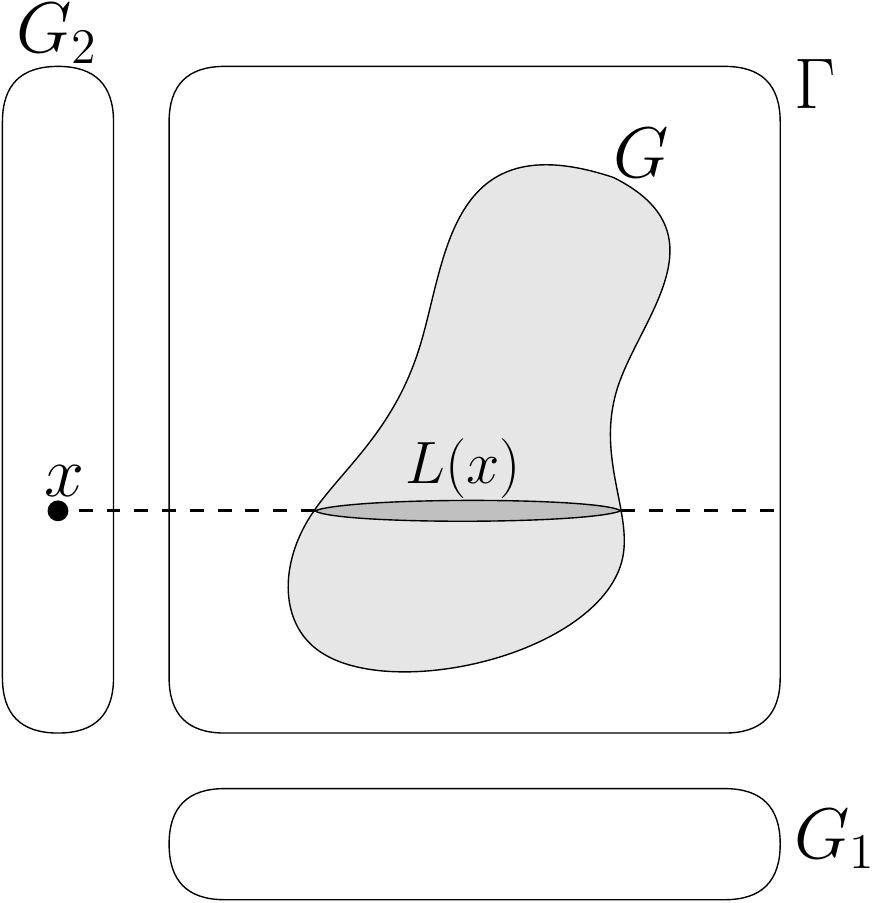}
	\end{minipage}
	\begin{minipage}{0.49\textwidth}
		\centering
		\includegraphics[width=0.5\textwidth]{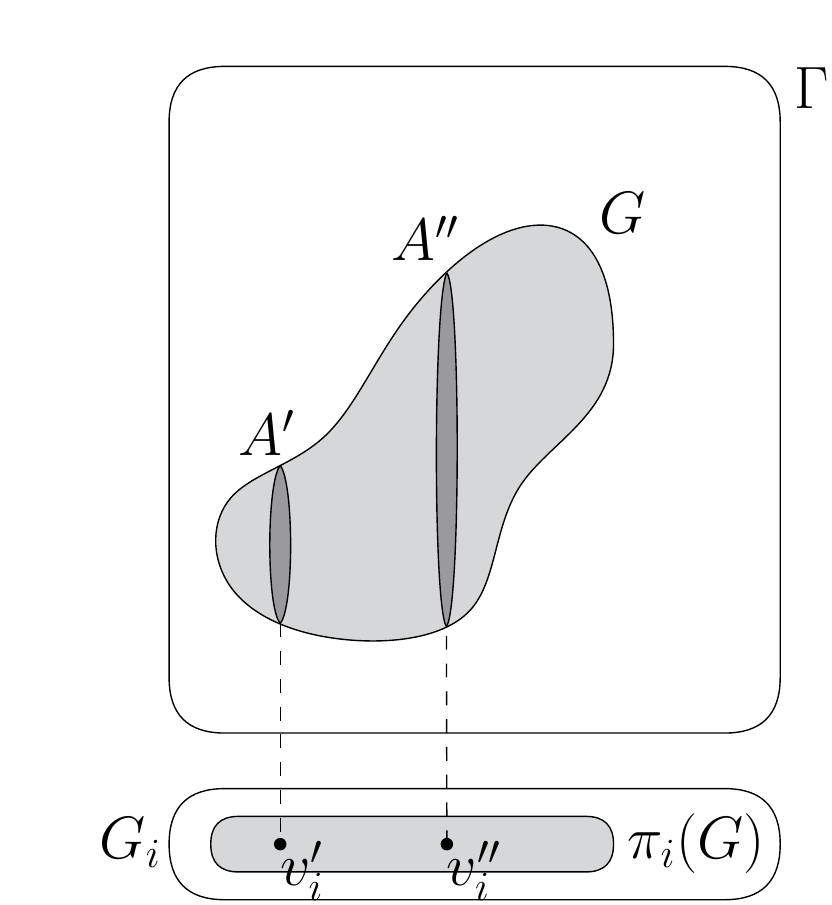}
	\end{minipage}
	
	\medskip
	\begin{minipage}{0.49\textwidth}
		\centering
		(a)
	\end{minipage}
	\begin{minipage}{0.49\textwidth}
		\centering
		(b)
	\end{minipage}
	\caption{\label{fig_layers} To the proofs of Lemma 
	\ref{lem_product_density} and Theorem
		\ref{thm_subgraphs_products1}.}
\end{figure}	

\subsubsection{Minors \cite{Die}}
A \emph{minor} of a graph $G$ is a graph $M$ obtained from a subgraph $G'$ of 
$G$ by
contracting some edges. Equivalently, $M$ is a minor of a connected graph $G$ 
if there exists a
partition of vertices of $G$ into connected subgraphs ${\mathcal P}=\{ 
P_1,\ldots,P_t\}$ and a
bijection $f: V(M)\rightarrow {\mathcal P}$ such that if $uv\in E(M)$ then 
there exists an edge of
$G$ running between the subgraphs $f(u)$ and $f(v)$ of $\mathcal P$, i.e., 
after contracting each subgraph $P_i\in {\mathcal P}$ into a single vertex  we 
will obtain a graph containing $M$ as a spanning subgraph. A class of graphs 
${\mathcal G}$ is called 
\emph{minor-closed} if for any graph
$G$ from ${\mathcal G}$ all minors of $G$ also belong to ${\mathcal G}$.

A \emph{minor-subproduct} of a Cartesian product $\Gamma:=\prod_{i=1}^mG_i$ of 
connected
graphs is a Cartesian product $M:=\prod_{i=1}^mM_{i}$, where $M_{i}$ is a minor 
of $G_{i}$ for all
$1\leq i\leq m$. Let ${\mathcal P}_{i}=\{ P^{i}_1,\ldots,P^{i}_{t_i}\}$ denote 
the partition of $G_{i}$
defining the minor $M_{i}$ and let ${\mathcal P}:={\mathcal 
P}_{1}\times\cdots\times{\mathcal
	P}_{m}$. Notice that $\mathcal P$ is a partition of the vertex set of the 
	Cartesian product
$\prod_{i=1}^mG_{i}$.

\subsection{VC-dimension and VC-density}

Let $\mcs$ be a family of subsets of a finite set $X=\{ e_1,\ldots, e_m\}$, 
i.e., $\mcs\subseteq 2^X$.
$\mcs$ can be viewed as a subset of vertices of the
$m$-dimensional hypercube $Q_m$. Denote by  $G(\mcs)$ the subgraph of $Q_m$ 
induced by the
vertices of $Q_m$ corresponding to the sets of $\mcs$;
$G(\mcs)$ is also called the \emph{1-inclusion graph} of $\mcs$ 
\cite{Hau,HaLiWa} (in \cite{GaGr},
such graphs were called {\it cubical graphs}).
Vice-versa, any induced subgraph $G=(V,E)$ of the hypercube $Q_m$ corresponds 
to a family of
subsets $\mcs(G)$ of $2^X$ with $|X|=m$ such that $G$
is the 1-inclusion graph of   $\mcs(G)$.

A subset $Y$ of $X$ is said to be \emph{shattered} by $\mcs$ if for any 
$Y'\subseteq Y$, there
exists a set   $S\in\mcs$ such that $S\cap Y=Y'$.
The \emph{Vapnik-Chervonenkis's dimension} \cite{VaCh,FuPa} $\vcd(\mcs)$ of 
$\mcs$ is the
cardinality of the largest subset of $X$ shattered by $\mcs$.
Viewing $Q_m$ as the $m$-fold Cartesian product $K_2\product \cdots\product
K_2$, the shattering
operation can be redefined in more graph-theoretical terms as follows.
For  $Y\subseteq X$ denote by $\Gamma_Y$ the Cartesian product of the factors of
$Q_m$ indexed by $e_i \in Y$ ($\Gamma_Y$ is a $|Y|$-dimensional cube).
Then a set $Y\subseteq X$ is shattered by $\mcs$ (or by $G(\mcs)$) if
$\pi_{\Gamma_Y}(G(\mcs))=\Gamma_Y$; see Fig. \ref{fig_shattering}, 
\ref{fig_P5_Q4}, and
\ref{fig_P5_Q3} for an illustration.

We continue with our main definitions of VC-dimension and VC-density for 
subgraphs of Cartesian
products of connected graphs. 	First, we define these notions with respect to 
subproducts.

\begin{definition} \label{vcdens} A subproduct $\Gamma':=\prod_{j=1}^kG'_{i_j}$ 
of a Cartesian
	product $\Gamma=\prod_{i=1}^m G_i$ is \emph{shattered} by $G$ if $F(v')\cap 
	V(G)\ne\emptyset$
	for any vertex $v'$
	of $\Gamma'$, i.e., if $\pi_{\Gamma'}(G)=\Gamma'$. The ({\it induced}) 
	\emph{VC-dimension}
	$\vcd(G)$ of $G$ with respect to the Cartesian product $\Gamma$ is the 
	largest number of
	non-trivial factors in a
	subproduct $\Gamma'$ of $\Gamma$ shattered by $G$. Equivalently, $\vcd(G)$ 
	is the largest
	dimension of a cube-subproduct of $\Gamma$ shattered by $G$ (since each 
	factor of $\Gamma'$
	is non-trivial).
	The \emph{VC-density} $\vcdens(G)$ of $G$ is  the largest density of a 
	subproduct $\Gamma'$
	shattered by $G$.
\end{definition}	

In the same vein, we consider now the notion of the VC-dimension and VC-density 
of subgraphs
$G$ of a Cartesian product $\Gamma$  with respect to minor subproducts.

\begin{definition} \label{minorvcdim}
	Let $M:=\prod_{i=1}^mM_{i}$ be a minor subproduct of $\Gamma=\prod_{i=1}^m 
	G_i$.  Let
	${\mathcal P}:={\mathcal P}_{1}\times\cdots\times{\mathcal P}_{m}$ be the 
	partition of $\Gamma$
	such that ${\mathcal P}_{i}=\{ P^{i}_1,\ldots,P^{i}_{t_i}\}$ is the 
	partition of $G_{i}$ defining the
	minor $M_{i}$, $i=1,\ldots,m$. The minor-subproduct $M$ is \emph{shattered} 
	by $G$
	if any set $P^{1}_{l_1}\times\cdots\times P^{m}_{l_m}$ of $\mathcal P$ 
	(where
	$l_1=1,\ldots,t_1,\ldots,\mbox{ and }l_m=1,\ldots,t_m$) contains a vertex 
	of $G$. The
	\emph{(minor) VC-dimension} $\vcd^*(G)$  of $G$ with respect to  $\Gamma$ is
	the largest number of non-trivial factors of a minor-subproduct $M$ of 
	$\Gamma$ shattered by
	$G$. The \emph{(minor) VC-density} $\vcdens^*(G)$  of $G$  is the largest 
	density $\dens(M)$ of
	a minor subproduct $M$ shattered by $G$.
\end{definition}

\begin{example} In Fig. \ref{example1}(1), we present a subgraph $G$ of the 
Cartesian product
	$\Gamma=G_1\product G_2$  of two graphs $G_1$ and $G_2$. In Fig.
	\ref{example1}(2), we provide
	a
	partition of $V(\Gamma)$ which induces a shattered minor isomorphic to
	the square $K_2\product
	K_2$. On the other hand, in Fig. \ref{example1}(3) we present a partition 
	of $\Gamma$ which
	induces
	the Cartesian product $P_3\product P_2$ of the paths of length 2 and
	length 1(see Fig.
	\ref{example1}(4)) and this product is not shattered by $G$.
\end{example}

\begin{figure}
	\begin{minipage}{0.49\textwidth}
		\centering
		\includegraphics[scale=0.3]{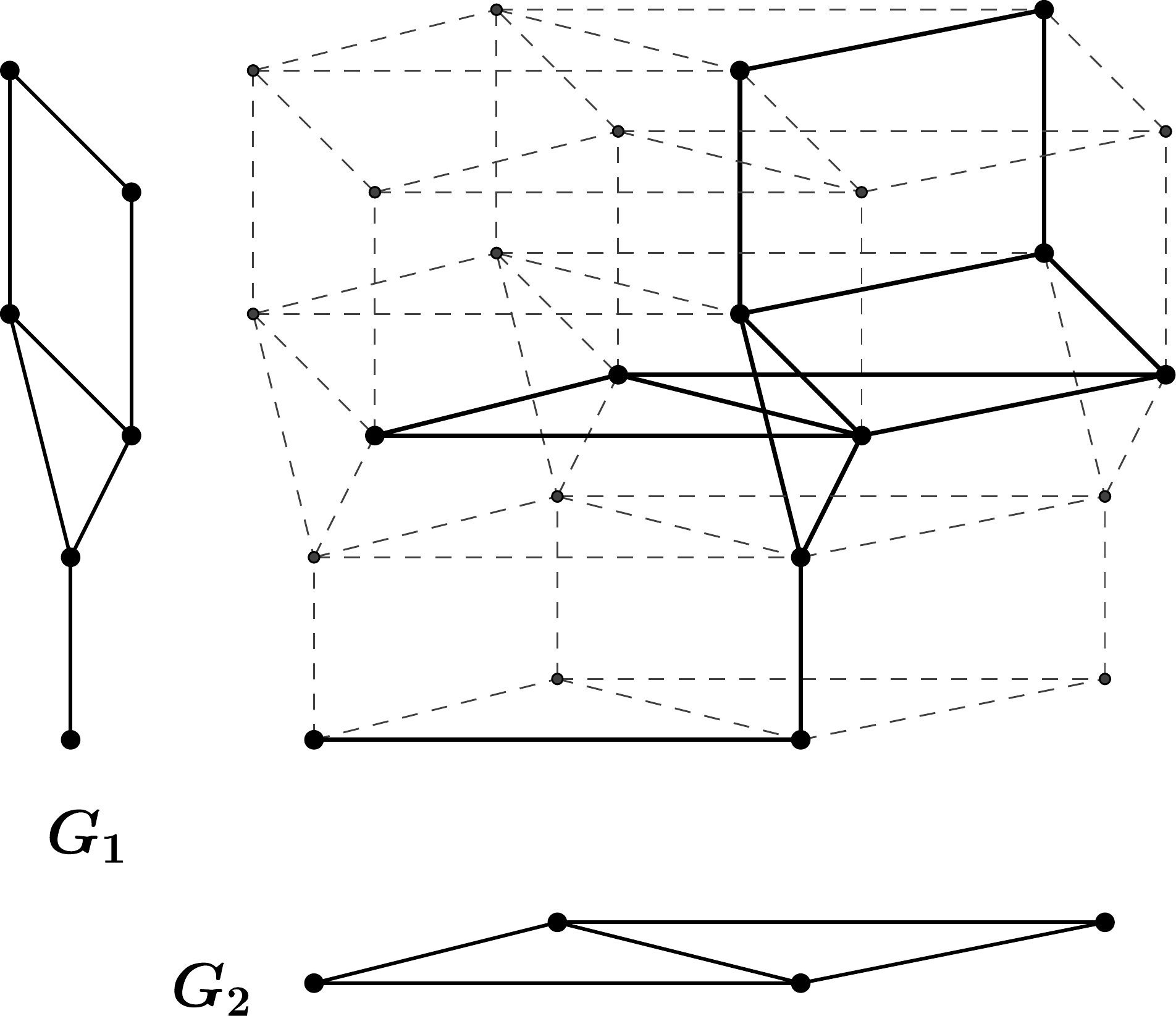}
	\end{minipage}
	\begin{minipage}{0.49\textwidth}
		\centering
		\includegraphics[scale=0.3]{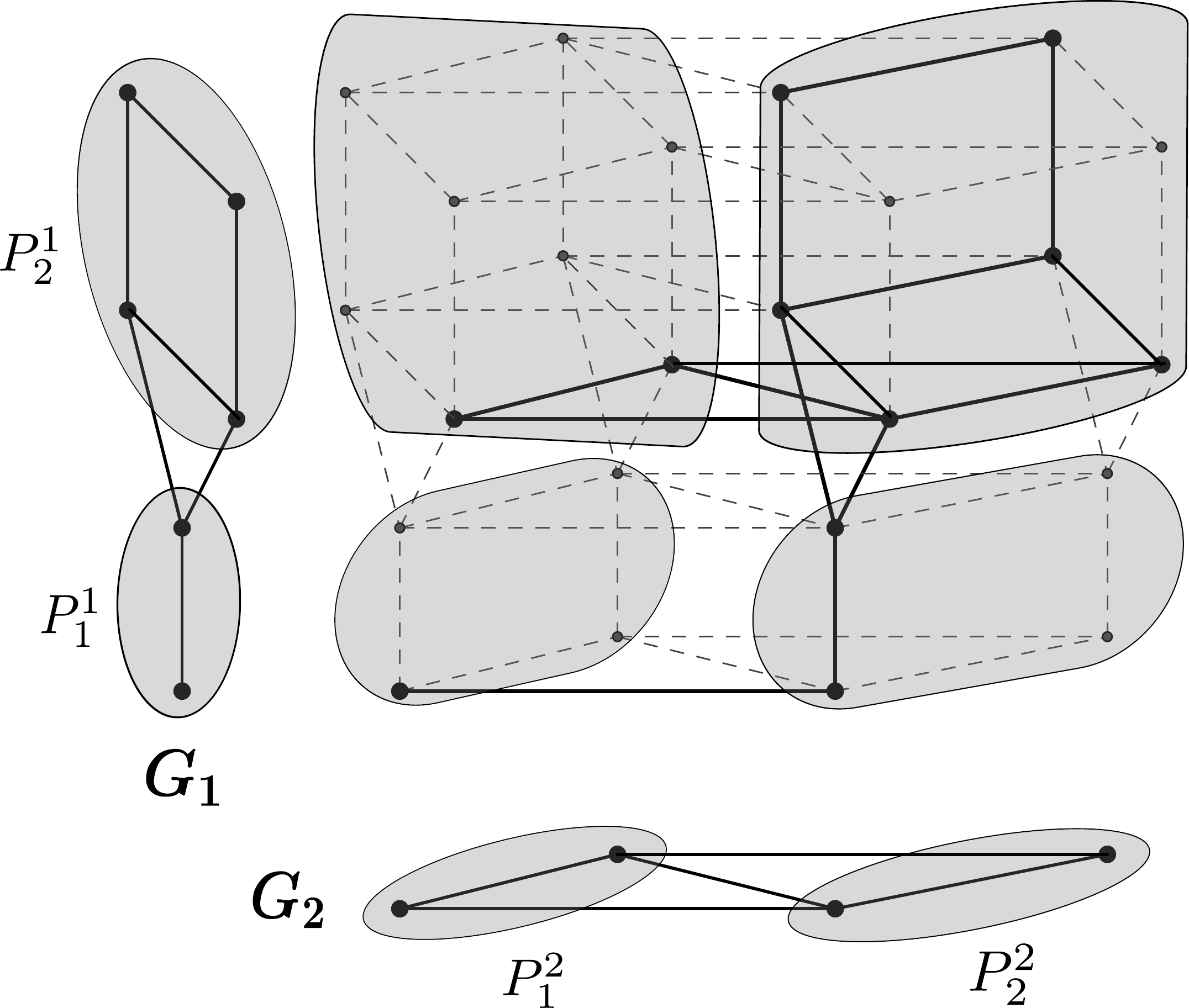}
	\end{minipage}
	
	\vskip1ex
	\begin{minipage}{0.49\textwidth}\centering{\bf (1)}\end{minipage}
	\begin{minipage}{0.49\textwidth}\centering{\bf (2)}\end{minipage}
	
	\vskip1ex
	\begin{minipage}{0.49\textwidth}
		\centering
		\includegraphics[scale=0.3]{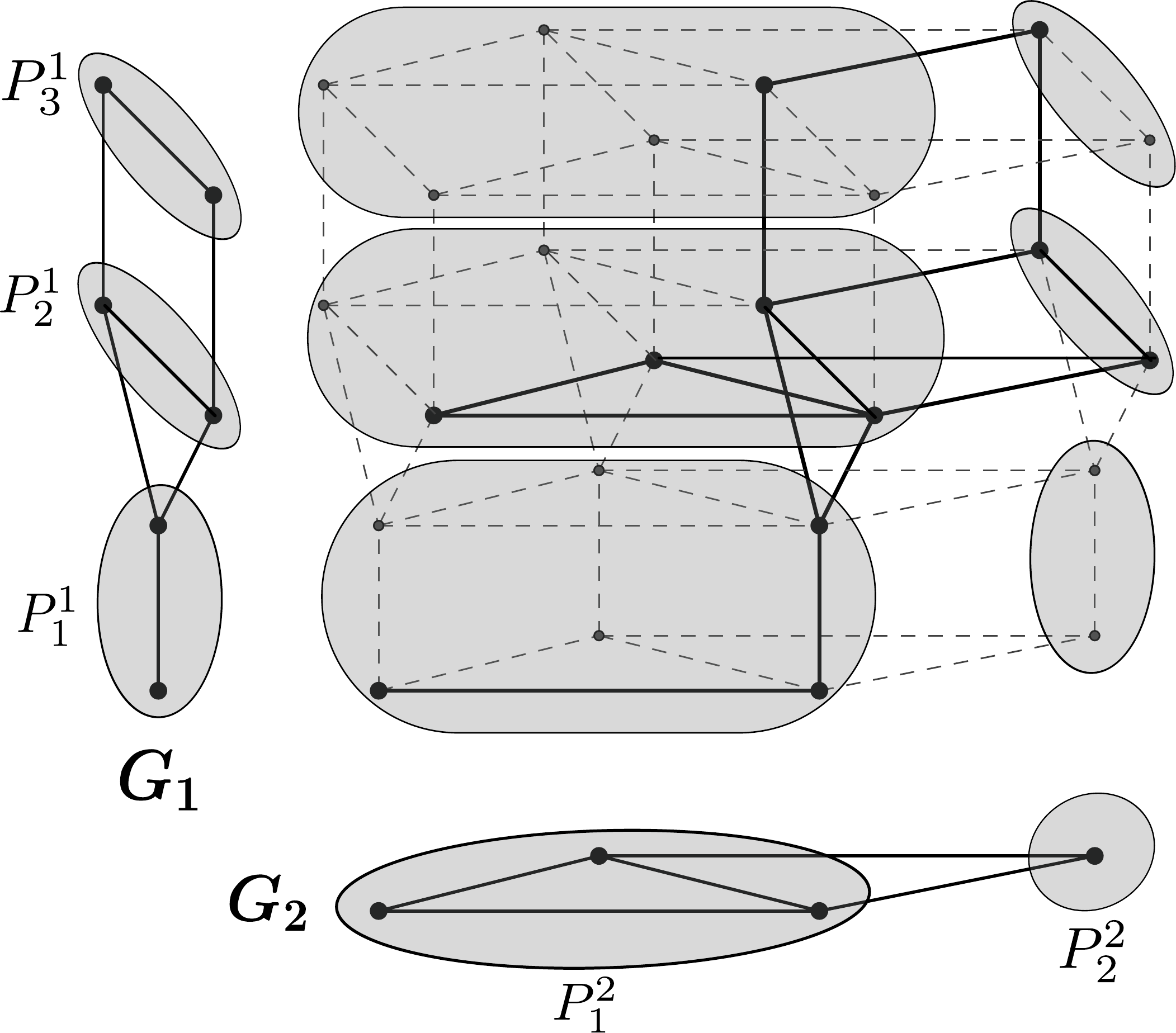}
	\end{minipage}
	\begin{minipage}{0.49\textwidth}
		\centering
		\includegraphics[scale=0.3]{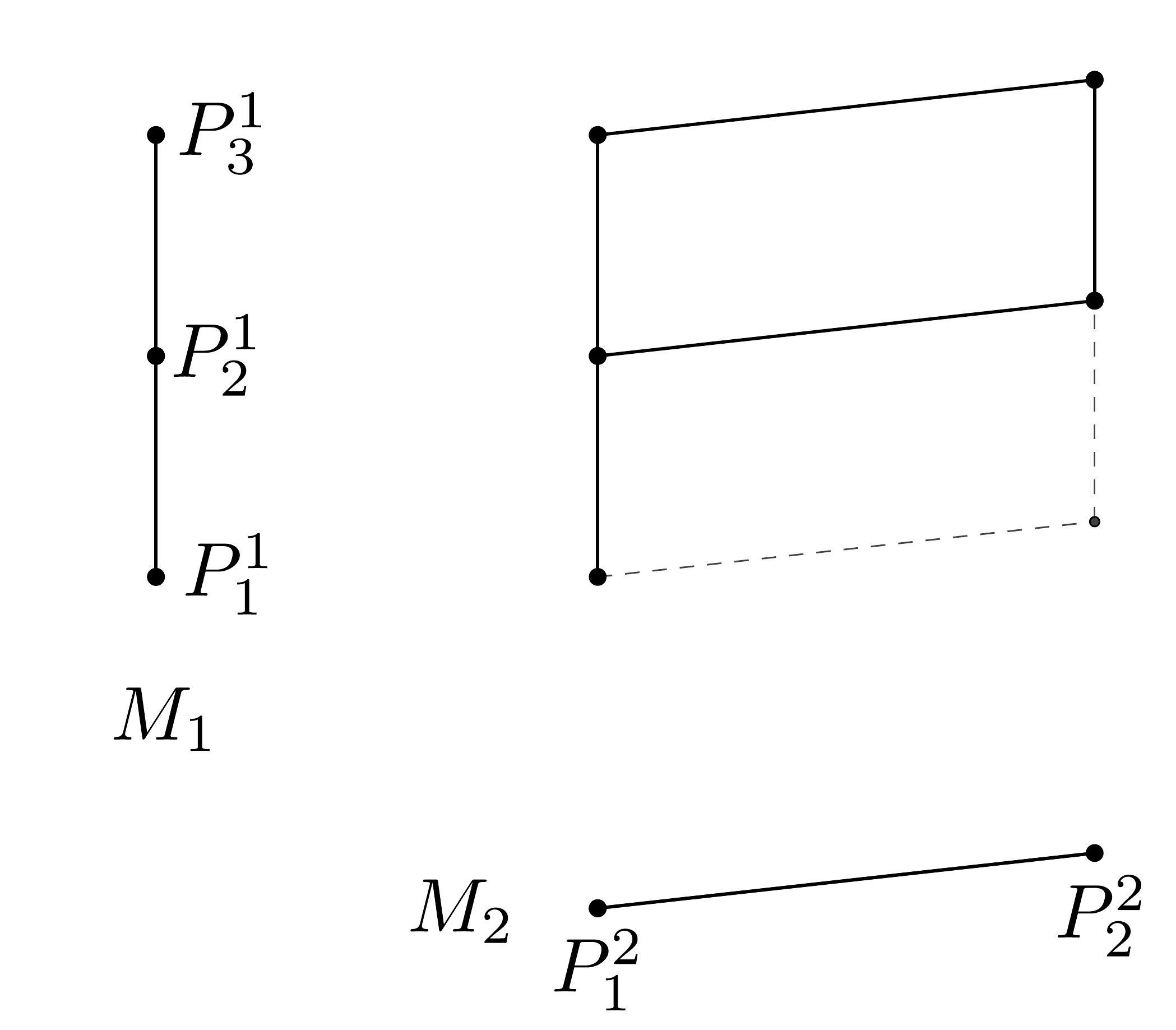}
	\end{minipage}
	
	\vskip1ex
	\begin{minipage}{0.49\textwidth}\centering{\bf (3)}\end{minipage}
	\begin{minipage}{0.49\textwidth}\centering{\bf (4)}\end{minipage}
	
	\caption{\label{example1} A subgraph $G$ of a Cartesian product of
	$\Gamma=G_1\product G_2$
		and two partitions  of $V(\Gamma)$, first  inducing a shattered by $G$ 
		minor-subproduct and
		second inducing a subproduct not shattered by $G$.}
\end{figure}

\begin{example} The inequalities $\vcd(G)\le \vcd^*(G)$ and $\vcdens(G)\le  
\vcdens^*(G)$ hold for
	any subgraph $G$ of $\Gamma$. In Fig.\ref{fig_vcd<vcd*}, we present a 
	simple example of a
	subgraph $G$ of the product of two 2-paths for which $\vcd(G)=1$ and 
	$\vcd^*(G)=2$.
\end{example}

\begin{figure}
	\begin{minipage}{0.49\textwidth}
		\centering
		\includegraphics[width=0.85\linewidth]{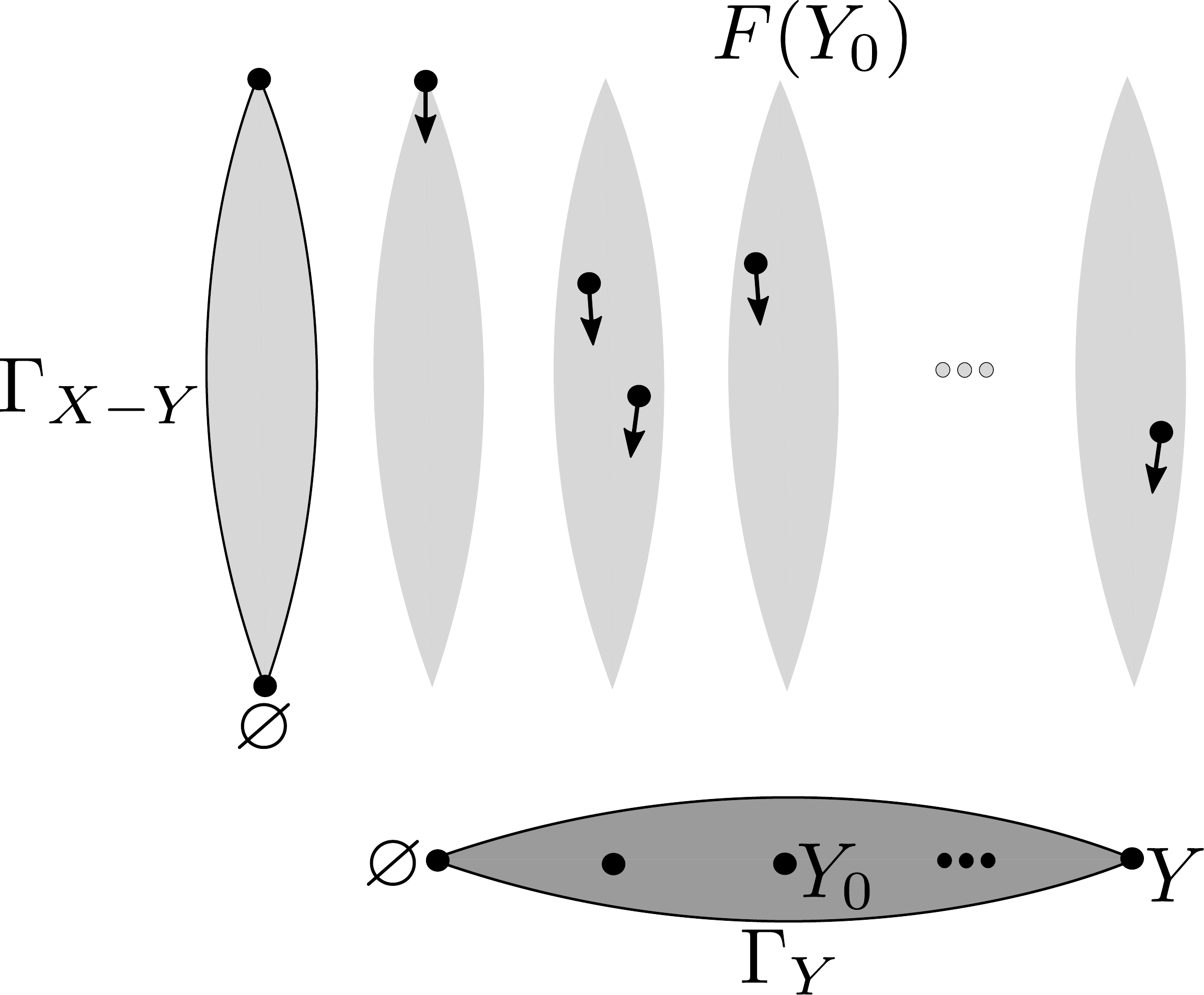}
	\end{minipage}
	\begin{minipage}{0.49\textwidth}
		\centering
		\includegraphics[width=0.7\linewidth]{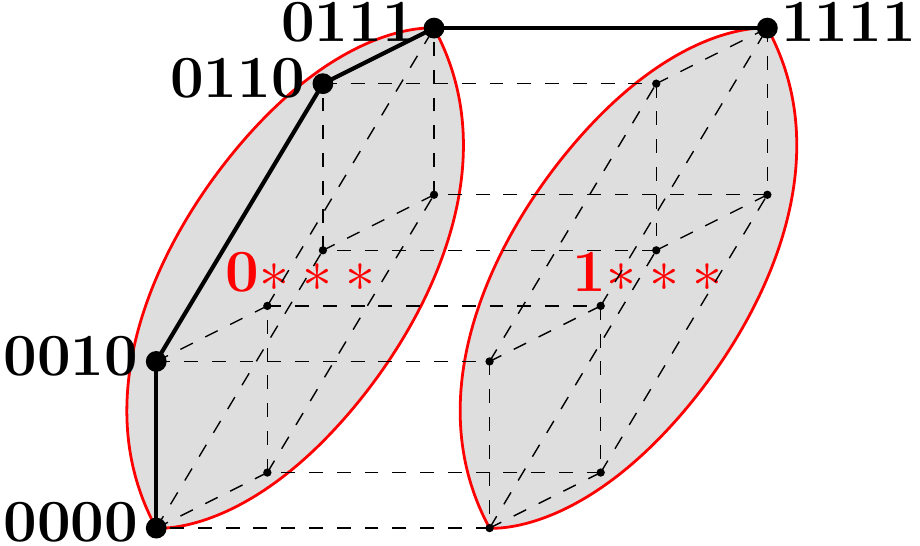}
	\end{minipage}
	
	\begin{minipage}{0.49\textwidth}\caption{\label{fig_shattering}
			A set $Y$ shattered by a set family $\mcs$ (the sets of $\mcs$ are 
			represented by black
			vertices in the gray fibers, see $F(Y_0)$, for example).
	}\end{minipage}
	\begin{minipage}{0.49\textwidth}\caption{\label{fig_P5_Q4} An embedding
	of $P_5$ in $Q_4$
			with VC-dimension $1$ ($\{0000,1000\}$ is shattered).}\end{minipage}
	
	\vskip2ex
	\begin{minipage}{0.49\textwidth}
		\centering
		\includegraphics[width=0.6\linewidth]{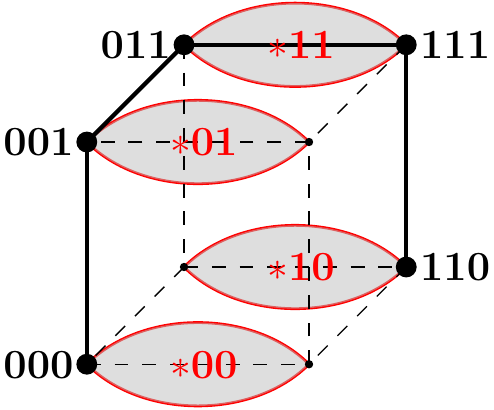}
	\end{minipage}
	\begin{minipage}{0.49\textwidth}
		\centering
		\includegraphics[width=0.7\linewidth]{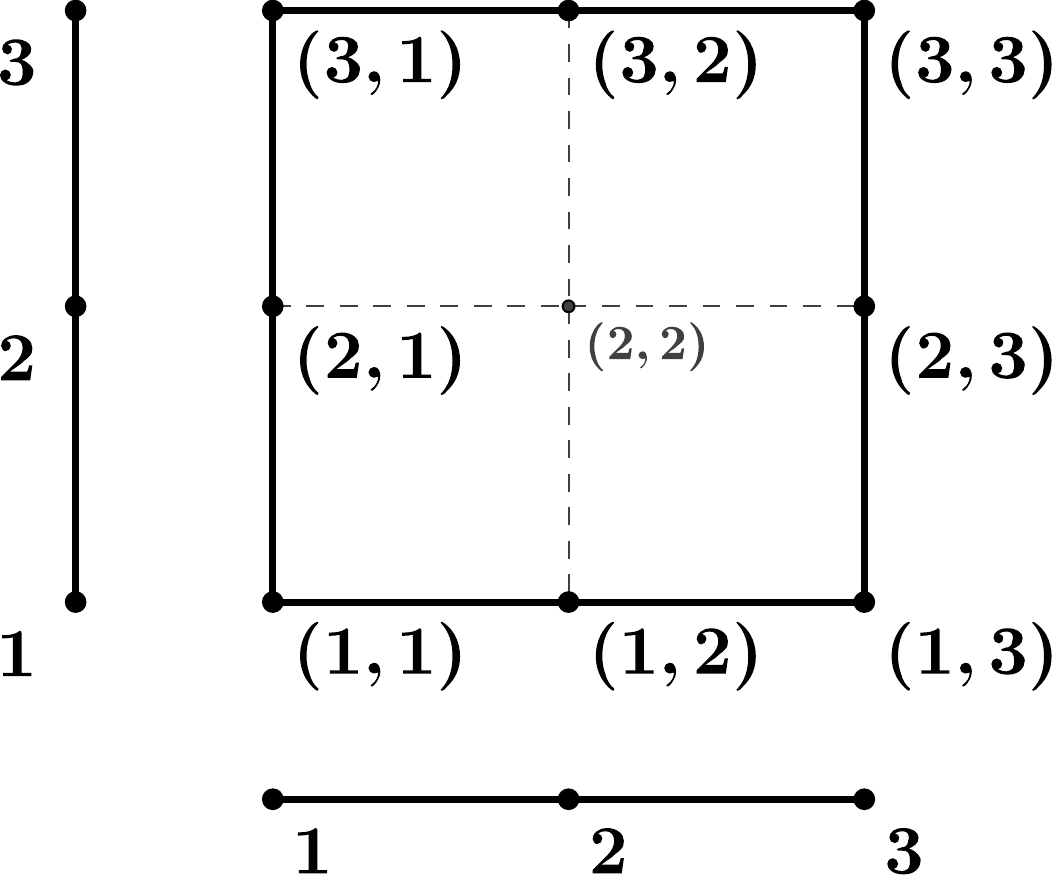}
	\end{minipage}
	
	\begin{minipage}{0.49\textwidth}\caption{\label{fig_P5_Q3} An embedding of 
	$P_5$ in $Q_3$
			with VC-dimension $2$ ($\{000,010,001,011\}$ is 
			shattered).}\end{minipage}
	\begin{minipage}{0.49\textwidth}\caption{\label{fig_vcd<vcd*} A
	subgraph $G$ of $P_3\product
			P_3$ that does not shatter any subproduct with $2$ factors but that 
			shatters a
			minor-subproduct induced by 
			$\mcp=\{\{1,2\},\{3\}\}\times\{\{1\},\{2,3\}\}$.}\end{minipage}
\end{figure}

\begin{remark}
	In case of  $m$-fold Cartesian products $F^m$ of a fixed graph $F$, 
	Cesa-Bianchi and Haussler
	\cite{CBHa} defined the notion of a $d$-dimensional projected cube, which 
	in this case  coincides
	with our notion of shattered cube-subfactor. 
\end{remark}

\begin{remark}
	As in the case of the classical VC-dimension, $\vcd(G)$ and $\vcd^*(G)$ are 
	defined with respect
	to an embedding of $G$ as a subgraph of the Cartesian product $\Pi_{i=1}^m 
	G_i$. For example, if
	$G$ is the path $P_5$ (with 5 vertices and 4 edges) and this path is 
	embedded in the 4-cube
	$Q_4$ such that the end-vertices of $P_5$ are $(0,0,0,0)$ and $(1,1,1,1)$, 
	then
	$\vcd^*(P_5)=\vcd(P_5)=1$. However, if $P_5$ is embedded in $Q_3$ (which 
	can be viewed as a
	face of $Q_4$) such the end-vertices of $P_5$ are $(0,0,0)$ and $(1,1,0)$, 
	then
	$\vcd^*(P_5)=\vcd(P_5)=2$. For illustration of the two embeddings of $P_5$ 
	see Figures
	\ref{fig_P5_Q4} and \ref{fig_P5_Q3}.
	
	In this paper, when we speak about $\vcd(G)$ or $\vcd^*(G)$ we assume that 
	an embedding of
	$G$ as an induced subgraph of a Cartesian product
	$G_1\product\cdots\product G_m$ is given. This is
	also essential from the computational point of view because already 
	recognizing if a graph $G$ is
	a subgraph of a hypercube is NP-complete \cite{AfPaPa}.
\end{remark}

\begin{remark}
	We present some motivation for the names ``VC-dimension'' and 
	``VC-density'' in the general
	setting of subgraphs of Cartesian products and for the way these concepts 
	have been defined.
	First notice that in case of subgraphs $G$ of hypercubes $Q_m$, i.e., 
	set-systems, the equality
	$\vcd(G)=\vcd^*(G)=\vcdens(G)=\vcdens^*(G)=\vcd(\mcs(G))$ holds and these 
	numbers  coincide
	with the dimension of the largest subcube $Q$ of $Q_m$ shattered by $G$. 
	This is because for
	hypercubes, i.e., Cartesian products of $K_2$, the notions of subproducts, 
	cube-subproducts,
	and minor-subproducts coincide. All these  dimensions coincide with the
	degrees of the vertices
	of $Q$ and thus coincide with the average degree $\mad(Q)$ of $Q$. This
	is one explanation
	why the VC-density of subgraphs of Cartesian products has been defined
	via the average
	degrees of shattered subproducts or shattered minor subproducts. Second, 
	Lemma
	\ref{lem_product_density} shows that subproducts are densest subgraphs of 
	Cartesian products.
	Therefore, one can expect that in case of subgraphs $G$ of Cartesian 
	products, the densest
	subproducts shattered by $G$ provide an upper bound for the density of $G$. 
	Third, for
	subgraphs of general products one cannot define $\vcdens(G)$ or 
	$\vcdens^*(G)$ as just the
	maximum number of factors in a shattered (minor) subproduct because the 
	factors in this
	subproduct may have completely different numbers of vertices, edges, or 
	average degrees.
	Finally, we use ``VC'' in the names because it concerns shattering, closely 
	related to classical
	VC-dimension.
\end{remark}

\section{Related work} \label{sect_original_method}  

\subsection{Subgraphs of hypercubes}	
In this subsection, we briefly review the inductive method of bounding the 
density of subgraphs
$G$ of hypercubes by the VC-dimension $d$ of  $\mcs=\mcs(G)$:
\begin{theorem}\cite[Lemma 2.4]{HaLiWa} \label{subgraphs_hypercubes} 
$\frac{|E(G)|}{|V(G)|}\le d.$
\end{theorem}

There are several  ways to prove this result (see \cite{Hau} for a proof using 
shifting operations, or
\cite{Har,Bezrukov} for edge-isoperimetric inequalities method), but our proofs 
in Section
\ref{products2} use the same idea as the inductive method we recall now. A 
similar proof also applies
to the classical {\it Sauer lemma} asserting that  sets families  of $\{ 
0,1\}^m$ of VC-dimension $d$
have size $O(m^d)$ (generalizations of Theorem \ref{subgraphs_hypercubes}
and of the Sauer lemma
were provided for subgraphs of Hamming graphs in \cite{HaLo,RuBaRu}). Both 
proofs are based on
the following fundamental lemma. For a finite set $X$ and $e\in X$, let 
$\mcs_e=\{ S'\subseteq
X\setminus \{ e\}: S'=S\cap X \text{ for some } S\in \mcs\}$ and $\mcs^e=\{ 
S'\subseteq X\setminus \{
e\}: S' \mbox{ and }  S'\cup \{ e\} \mbox{ belong to } \mcs\}.$

\begin{lemma} \label{VC-dim-prop} $|\mcs|=|\mcs_e|+|\mcs^e|$,  $\vcd(\mcs_e)\le 
d$, and
	$\vcd(\mcs^e)\le d-1$.
\end{lemma}

The proof of the inequality $|E(G)|\le d\cdot |V(G)|$ in  Theorem 
\ref{subgraphs_hypercubes}
provided by Haussler, Littlestone, and Warmuth \cite{HaLiWa} is by induction 
using Lemma
\ref{VC-dim-prop}. 	 For $e\in X$, denote by $G_e$ and $G^e$ the subgraphs of 
$2^{X-\{ e\}}$
induced by  $\mcs_e$ and $\mcs^e$.  Then $|E(G_e)|\le d|V(G_e)|=d|\mcs_e|$ and 
$E(G^e)\le
(d-1)|V(G^e)|=(d-1)|\mcs^e|$ by Lemma \ref{VC-dim-prop} and induction 
hypothesis. The graph
$G_e$ is obtained from $G$ by contracting  the set $F$ of edges of type $e$ of 
$G$. The vertex set
of $G^e$ is in bijection with $F$ and two vertices of $G^e$ are adjacent iff 
the corresponding edges
of $F$ belong to a common square. By  Lemma \ref{VC-dim-prop},
$|V(G)|=|\mcs|=|\mcs_e|+|\mcs^e|=|V(G_e)|+|V(G^e)|$. The edges of $G$ which 
lead to loops or
multiple edges of $G_e$ are the edges of $G^e$ and of $F$, hence
$|E(G)|\leq|E(G_e)|+|E(G^e)|+|F|=|E(G_e)|+|E(G^e)|+|V(G^e)|$. From this 
(in)equality and after some
calculation, one deduce that $\frac{|E(G)|}{|V(G)|}\le d$.

\subsection{Other notions of VC-dimension in graphs}
We continue with a brief survey of other existing notions of VC-dimension in 
graphs.  Haussler and
Welzl \cite{HaWe} defined the VC-dimension of a graph $G=(V,E)$ as the 
VC-dimension of
the set family of closed neighborhoods of vertices of $G$. It was shown in 
\cite{HaWe} that this
VC-dimension of  planar graphs is at most 4. Anthony, Brightwell, and Cooper 
\cite{AnBrCo}
proved that this VC-dimension is at most $d$ if $G$ does not contain $K_{d+1}$ 
as a minor (they
also investigated this notion of VC-dimension for random graphs). These two 
results have
been extended in \cite[Proposition 1 \& Remark 1]{ChEsVa} to the families of 
closed balls of any fixed
radius. Kranakis et al. \cite{KrKrRuUrWo} considered the VC-dimension for other
natural families of sets of a graph: the families induced by trees, connected 
subgraphs, paths,
cliques, stars, etc.). They investigated the complexity issues for computing 
these
VC-dimensions and for some of them they presented upper bounds in terms of other
graph-parameters. The concept of VC-dimension of the family of shortest paths 
(in graphs with
unique
shortest paths) was exploited in \cite{AbDeFiGoWe} to improve the time bounds 
of query algorithms
for point-to-point shortest path problems in real world networks (in
particular, for road networks).

\section{Our results} \label{our_results}
The main purpose of this paper is to generalize the density results about 
subgraphs of hypercubes
to subgraphs $G$ of Cartesian products of connected graphs $G_1,\ldots,G_m$. 
Namely, we will
prove the following two results:

\begin{theorem}\label{thm_subgraphs_products1} Let $G$ be a subgraph of a 
Cartesian product
	$G_1\product\cdots\product G_m$ of connected graphs $G_1,\ldots, G_m$
	and let
	$\beta:=\ceil{\max\{\mad(G_1),\ldots,\textsl{}\mad(G_m)\}}$,
	$\beta_0=\ceil{\max\{\mad(\pi_1(G)),\ldots,\mad(\pi_m(G))\}}$. Then
	$$\frac{|E(G)|}{|V(G)|}\leq \beta_0\log |V(G)|\le \beta\log |V(G)|.$$
\end{theorem}

For a graph $H$, denote by ${\mathcal G}(H)$ the set of all finite 
$H$-minor-free graphs, i.e., graphs
not having $H$ as a minor. By results of Mader, Kostochka, and Thomason (see 
\cite[Chapter
8.2]{Die}), any graph $G$ with average degree $d(G)\ge cr\sqrt{\log r}$  has 
the complete graph
$K_r$ on $r$ vertices as a minor. Therefore, for each graph $H$ on $r$ 
vertices, there exists a
constant $\mu(H)\le cr\sqrt{\log r}$ such that all graphs $G$ from ${\mathcal 
G}(H)$ have average
degree $d(G)\le \mu (H)$. In case when the factors of the product $\Gamma$ 
belong to ${\mathcal
	G}(H)$ the following result, generalizing Theorem 
	\ref{subgraphs_hypercubes} and sharpening
Theorem \ref{thm_subgraphs_products1}, holds:

\begin{theorem}\label{thm_subgraphs_products} Let $H$ be a graph and let $G$ be 
a subgraph of
	a Cartesian product $\Gamma=G_1\product\cdots\product G_m$ of connected
	graphs
	$G_1,\ldots,G_m$
	from ${\mathcal G}(H)$. Then
	$$
        \frac{|E(G)|}{|V(G)|} \le \mu (H)\cdot\vcd^*(G) 
                              \le \mu(H) \cdot\log |V(G)|.
    $$
\end{theorem}

We conjecture that in fact a stronger result holds:

\begin{conjecture} \label{conj_E/V<vcdens} Let $G$ be a subgraph of the 
Cartesian product
	$G_1\product\cdots\product G_m$. Then  $\frac{|E(G)|}{|V(G)|}\le
	\vcdens^*(G).$
\end{conjecture}

A partial evidence for this conjecture is Lemma \ref{lem_vcdens<alpha.vcdim} 
below showing that
$\vcdens^*(G) \leq \frac{\mu(H)}{2}\cdot\vcd^*(G)$. 	

As a direct consequence of \cite[Lemma 3.1]{AlTa} and Theorem 
\ref{thm_subgraphs_products}, we
obtain the following:

\begin{corollary}\label{corol_orientation}
	If $G$ is a subgraph of a Cartesian product of connected graphs from 
	${\mathcal G}(H)$ and $d
	:= \vcd^*(G)$, then $G$ has an orientation $D$ in which every outdegree is 
	at most $d\mu(H)$.
	
\end{corollary}

The proof of Theorem \ref{thm_subgraphs_products1} is given in Section 
\ref{sect_products1}. The
proof of Theorem \ref{thm_subgraphs_products} occupies Section \ref{products2}, 
where we also
present basic properties of the VC-dimension and VC-density of subgraphs of 
Cartesian products.

\section{Proof of Theorem \ref{thm_subgraphs_products1}}\label{sect_products1}

Let $n:=|V(G)|$. Since $\dens(G'_i)\le \dens(G_i)$ for any subgraph $G'_i$ of 
$G_i$ for
$i=1,\ldots,m$,
$$\ceil{\max\{\dens(\pi_1(G)),\ldots,\dens(\pi_m(G))\}}\le 
\ceil{\max\{\dens(G_1),\ldots,\dens(G_m)\}}$$
holds, establishing the second inequality $\beta_0\log |V(G)|\le \beta\log 
|V(G)|$ (recall that
$\pi_i(G)$ is the projection of $G$ on factor $G_i$). We prove the inequality
$$\frac{|E(G)|}{|V(G)|}\leq \beta_0\log
|V(G)|=\ceil{\max\{\mad(\pi_1(G)),\ldots,\mad(\pi_m(G))\}}\cdot\log n$$
by induction on $n$. If $n=1$, then we are obviously done. So suppose that 
$n\ge 2$. Then the
projection $\pi_i(G)$ on some factor $G_i$ contains at least two vertices. From
Lemma~\ref{prop:2vertexmad}, it follows that $\pi_i(G)$ has two vertices $v'_i$ 
and $v''_i$ of degree
at most $\ceil{\mad(\pi_i(G))}\le \beta_0$. Denote by $A'$ (resp. $A''$) the 
set of all vertices of the
graph $G$ having $v'_i$ (resp. $v''_i$) as their $i$th coordinate (see 
Fig.\ref{fig_layers}(b)).
Since $v'_i,v''_i\in V(\pi_i(G))$, both $A'$ and $A''$ are nonempty. At least 
one of the sets $A',A''$
contains at most $n/2$ vertices, say $|A'|\le \frac{n}{2}$. Set 
$B:=V(G)\setminus A'$ and denote by
$G'$ and $G''$ the subgraphs of $G$ induced by the sets $A'$ and $B$. Let 
$\beta'_0:=\ceil{\max\{
	\mad(\pi_1(G')),\ldots,\mad(\pi_m(G'))\}}$ and $\beta''_0:=\ceil{\max\{
	\mad(\pi_1(G'')),\ldots,\mad(\pi_m(G''))\}}$. Since $\pi_j(G')$ and 
	$\pi_j(G'')$ are subgraphs of
$\pi_i(G)$,
$j=1,\ldots,m$, $\beta'_0,\beta''_0\le \beta_0$. By induction assumption, 
$|E(G')|\leq
\beta'_0|A'|\log|A'|\le \beta_0|A'|\log|A'|$ and $|E(G'')|\leq 
\beta''_0|B|\log|B|\le  \beta_0|B|\log|B|$.
Since $A'$ and $B$ partition $V(G)$, $E(G)$ is the disjoint union of 
$E(G'),E(G'')$, and $E(A',B)$,
where $E(A',B)$ is the set of edges of $G$ with one end in $A'$ and another end 
in $B$. Since any
vertex of $A'$ has at most $d(v'_i)\le\ceil{\mad(\pi_i(G))}\le \beta_0$ 
incident edges in $E(A',B)$, we
obtain $|E(A',B)|\le \beta_0\cdot|A'|=\beta_0|A'|\log2$. Consequently,
$$
\begin{aligned}
|E(G)| &=|E(G')| + |E(G'')| +|E(A',B)|									\\
&\leq \beta_0|A'|\log|A'| + \beta_0|B|\log|B| + \beta_0|A'|\log2 \\
&= \beta_0|A'|\log(2|A'|) + \beta_0|B|\log|B|	\\
&\leq \beta_0(|A'|+|B|)\log n = \beta_0 n \log n.
\end{aligned}
$$
\qed

\section{Proof of Theorem \ref{thm_subgraphs_products}}\label{products2}

The proof of Theorem \ref{thm_subgraphs_products} goes along the lines of the 
proof of Theorem
\ref{subgraphs_hypercubes} of \cite{HaLiWa}
but is technically more involved.  The roadmap of the proof is as follows.
The first three results (Lemmas \ref{VC-dim-prop1}, 
\ref{lem_vcdens<alpha.vcdim}, and
\ref{VC-dim-prop2}) present the elementary properties of VC-dimension and 
VC-density
and some relationships between them. To prove an analog of Lemma 
\ref{VC-dim-prop} we have to
extend
to subgraphs of products of graphs the operators $\mcs_e$ and $\mcs^e$ defined 
for set systems
$\mcs$. For this, we pick  two adjacent
vertices $u$ and $v$ of some factor $G_i$ and define the graphs $G_{uv}$ and 
$G^{uv}_c$. The
graph $G_{uv}$ is obtained from $G$ by contracting every edge of type $uv$.
$G_{uv}$ is a subgraph of the Cartesian product having the same factors as
$\Gamma=G_1\product\cdots\product G_m$, only the $i$th factor is $G_i$ in
which the edge
$uv$ is contracted. The definition of the graph $G^{uv}_c$ is more involved and 
is given below.
Then we prove the analogues of the inequalities $\vcd(\mcs_e)\le d$ and
$\vcd(\mcs^e)\le d-1$ for graphs $G_{uv}$ and $G^{uv}_c$ (Lemmas 
\ref{prop_VC(G_uv)} and
\ref{prop_VC(G^uv)}). Finally, we have to obtain analogues of the equality
$|\mcs|=|\mcs_e|+|\mcs^e|$ and of the inequality $|E(G)|\leq 
|E(G_e)|+|E(G^e)|+|V(G^e)|$ to the
graphs $G, G_{uv},$ and $G^{uv}_c$. This is done in Lemma 
\ref{lem_counting_V_and_E}.
To proceed by induction and to prove that  $\frac{|E(G)|}{|V(G)|}\le 
\mu(H)\cdot\vcd^*(G)$, in Lemma
\ref{lem_G_c^uv} we show that
$\frac{|V(G^{uv})|}{|V(G^{uv}_c)|}\le |N|\le \mu(H)-1$, where $N$ is the set of 
the common neighbors
of $u$ and $v$. This is the case if one of the vertices $u$ or $v$ has
degree $\le \mu(H)$ (in our case, such a vertex always exists since each factor 
$G_i$ is
$\mu(H)$-degenerated).

\subsection{Properties of VC-dimension and 
VC-density}\label{sect_properties_VC-dim}

We continue with some basic properties of minor and induced VC-dimensions for 
products of
arbitrary connected graphs and extend Lemma \ref{VC-dim-prop}. In all these 
results,
we suppose that $G$ is a  subgraph of the Cartesian product 
$\Gamma:=\prod_{i=1}^m
G_i=G_1\product\cdots\product G_m$ of connected graphs  $G_1,\ldots,G_m$
and that $G$ has $n$
vertices.
Since shattering in the definition of $\vcd(G)$ and $\vcdens(G)$ is respect to 
subproducts and
since all subproducts are minor-subproducts,  we immediately obtain:

\begin{lemma} \label{VC-dim-prop1} $\vcd(G)\le \vcd^*(G)$ and $\vcdens(G)\le 
\vcdens^*(G).$
\end{lemma}

The following lemma justifies in part the formulation of Conjecture 
\ref{conj_E/V<vcdens}:
\begin{lemma} \label{lem_vcdens<alpha.vcdim} $\vcdens^*(G) \leq 
\frac{\mu(H)}{2}\cdot\vcd^*(G)$.
\end{lemma}
\begin{proof}
	Let $M=M_{i_1} \product\cdots\product M_{i_k}$ be a minor-subproduct of
	$\Gamma$ such that
	$\frac{|E(M)|}{|V(M)|}=\vcdens^*(G)$, and let $M'=M_{j_1}'
	\product\cdots\product M_{j_h}'$ be a
	minor-subproduct such that $h=\vcd^*(G)$. First, notice two things: (1) 
	$h\geq k$ because, by
	definition, we chose $M$ with no trivial factors (so if $k$ was greater 
	than $h$ we would have
	taken, at least, $M'=M$); (2) a simple counting of the vertices and edges 
	of a Cartesian product
	$\Gamma'=G'_1 \product\cdots\product G'_m$ shows that
	$\frac{|E(\Gamma')|}{|V(\Gamma')|} =
	\frac{\sum_{i=1}^{m} |E(G'_i)| \cdot \prod_{j \neq i} 
	|V(G'_j)|}{\prod_{i=1}^{m} |V(G'_i)|} =
	\sum_{i=1}^{m}\frac{|E(G'_i)|}{|V(G'_i)|}$.
	We then have 
	$\frac{|E(M)|}{|V(M)|}=\sum_{l=1}^k\frac{|E(M_{i_l})|}{|V(M_{i_l})|}$. 
	Since we defined
	$\Gamma$ is a product of $H$-minor free graphs, for all 
	$l\in\{1,\ldots,k\}$,
	$\frac{|E(M_{i_l})|}{|V(M_{i_l})|}=\frac{1}{2}d(M_{i_l})\leq\frac{\mu(H)}{2}$.
	 Then,
	$\frac{|E(M)|}{|V(M)|}=\vcdens^*(G)\leq \frac{\mu(H)}{2}\cdot k \leq 
	\frac{\mu(H)}{2}\cdot
	h=\frac{\mu(H)}{2}\cdot\vcd^*(G)$.
\end{proof}

\begin{lemma} \label{VC-dim-prop2} 
	$\vcd^*(G)\le \log n$.
\end{lemma}

\begin{proof} Let  $M=M_{1}\product\cdots\product M_{m}$ be a minor
subproduct of $\prod_{i=1}^m
	G_i$ shattered by $G$ with $k=\vcd^*(G)$ non-trivial factors. We want to 
	show that $k\le \log
	n$.
	Let ${\mathcal P}_{i}=\{ P^{i}_1,\ldots,P^{i}_{t_i}\}$ denote the partition 
	of $G_{i}$ defining the
	minor
	$M_{i}$. Since every non-trivial factor of $M$ contains at least two 
	vertices, ${\mathcal
		P}:={\mathcal
		P}_{1}\times\cdots\times {\mathcal P}_{m}$ contains at least $2^k$ 
		parts. By the definition of
	shattering, any of those parts $P^{1}_{l_1}\times\cdots\times P^{m}_{l_m}$ 
	(where
	$l_1=1,\ldots,t_1,\ldots,\mbox{ and } l_m=1,\ldots,t_m$) contains a vertex 
	$x$ of $G$. Since
	$\mcp$ is
	a partition of $\Gamma$, two vertices belonging to different parts have to 
	be different.
	Consequently, $G$ contains at least $2^k$ vertices.
\end{proof}

We continue with the extension to subgraphs of Cartesian  products of the 
operators $\mcs_e$ and
$\mcs^e$ defined for set systems $\mcs$. In case of set systems $\mcs$, the 
1-inclusion graph
$G(\mcs)$ is an induced subgraph of the product of $K_2$'s. Then $e$ 
corresponds to a factor of
this product and $G(\mcs_e)$ can be viewed as the image of $G$ in  the product  
of $K_2$'s
where the whole factor corresponding to $e$ was contracted.  In case when the 
factors are arbitrary
graphs, contracting a whole factor of the product would be too rough. So, let 
$u$ and $v$ be two
adjacent vertices of some factor $G_i$. Let $N$ denote the set of common 
neighbors of $u$ and
$v$ in $G_i$. Let  $\widehat G_i$ be the graph obtained from $G_i$ by 
contracting the edge $uv$,
namely, the graph in which the edge $uv$ is replaced by a vertex $w$ and every 
edge $xu$ and/or
$xv$ of $G_i$ is replaced by a single new edge $xw$; thus $\widehat G_i $ does 
not contain loops
and multiple edges. Let $\widetilde G_i$ be the graph which is a star having  
as the central vertex a
vertex  $\widetilde{w}$ corresponding to the  edge  $uv$  and as the set 
$\widetilde{N}$ of leaves
the vertices $\widetilde{x}$ corresponding to vertices $x$ of $N$ (i.e., such 
that $xuv$ is a triangle of $G_i$); the edges of  $\widetilde{G_i}$ are all 
pairs of the form $\widetilde{w}\widetilde{x}$.

\begin{figure}
	\begin{minipage}{0.49\textwidth}
		\centering
		\includegraphics[scale=0.3]{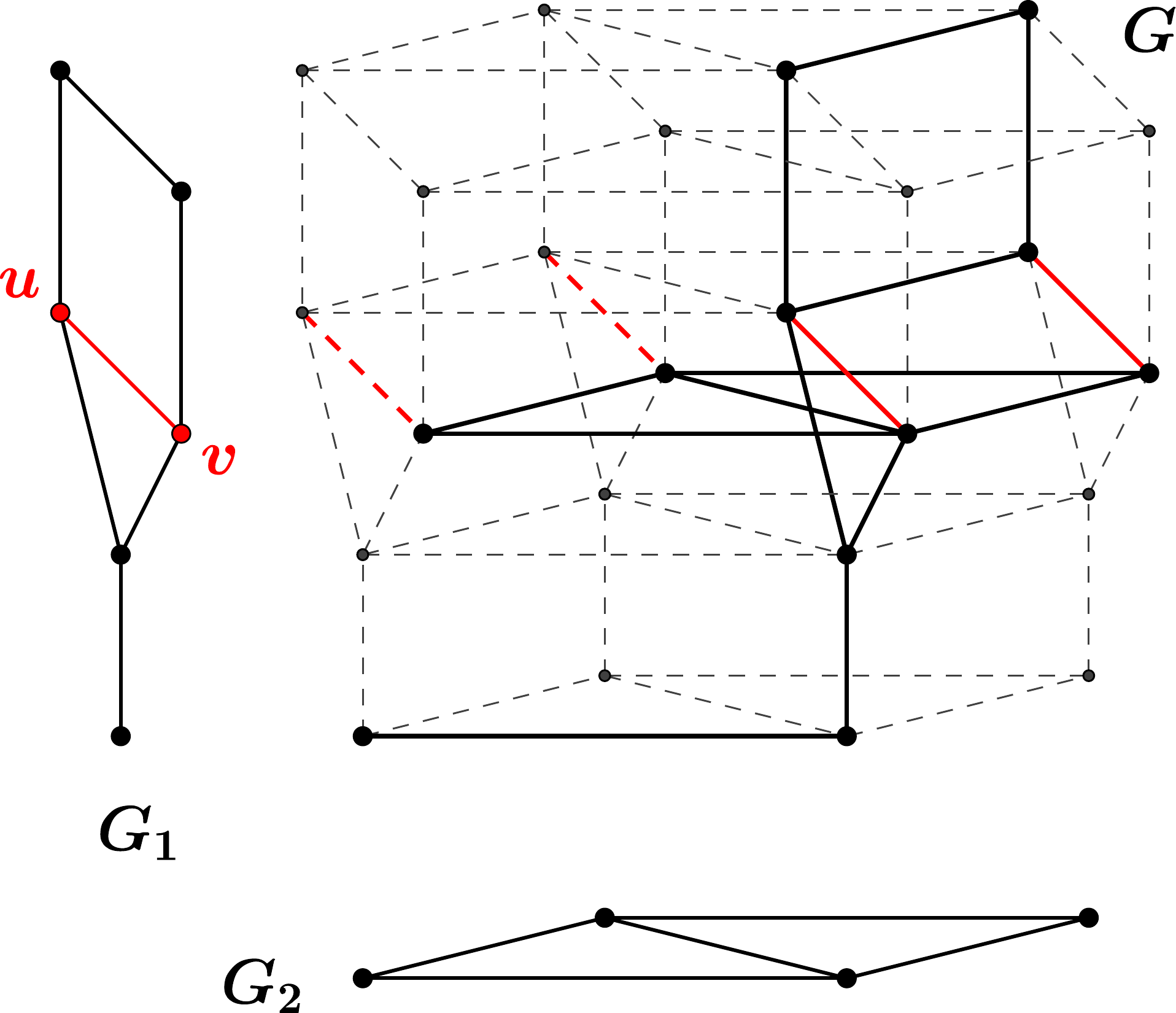}
	\end{minipage}
	\begin{minipage}{0.49\textwidth}
		\centering
		\includegraphics[scale=0.3]{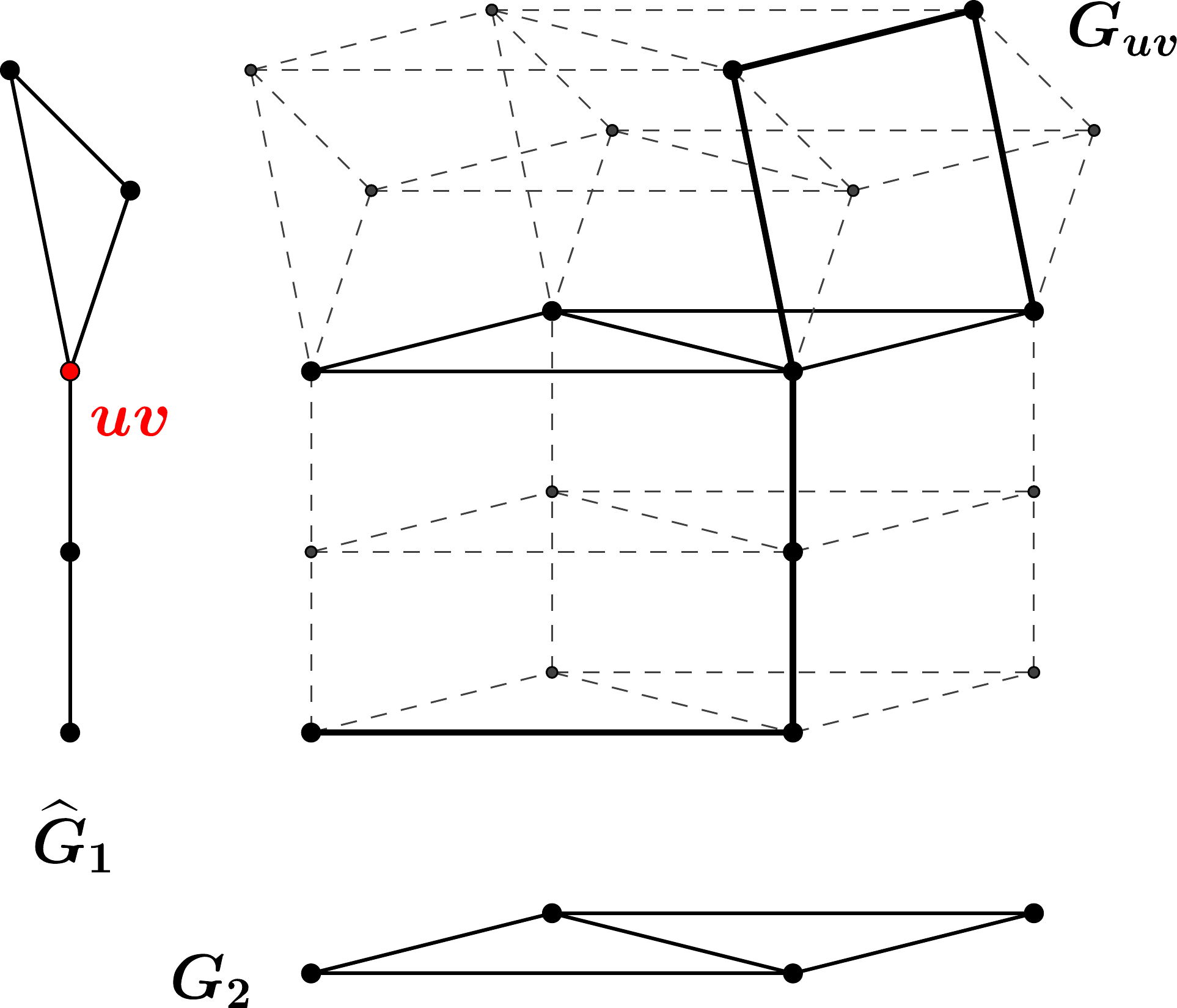}
	\end{minipage}
	
	\begin{minipage}{0.49\textwidth}
		\centering
		\includegraphics[scale=0.3]{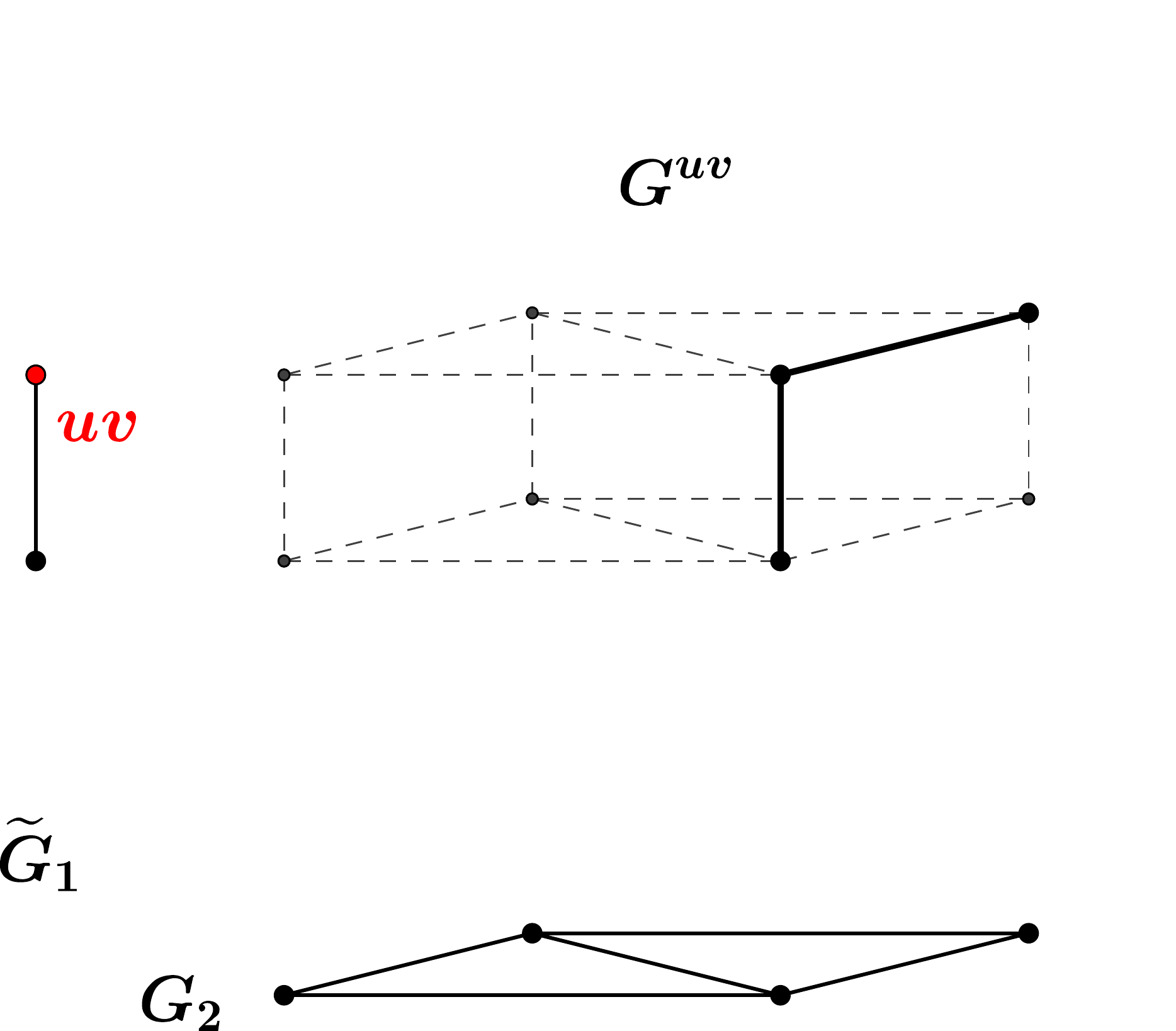}
	\end{minipage}
	\begin{minipage}{0.49\textwidth}
		\centering
		\includegraphics[scale=0.3]{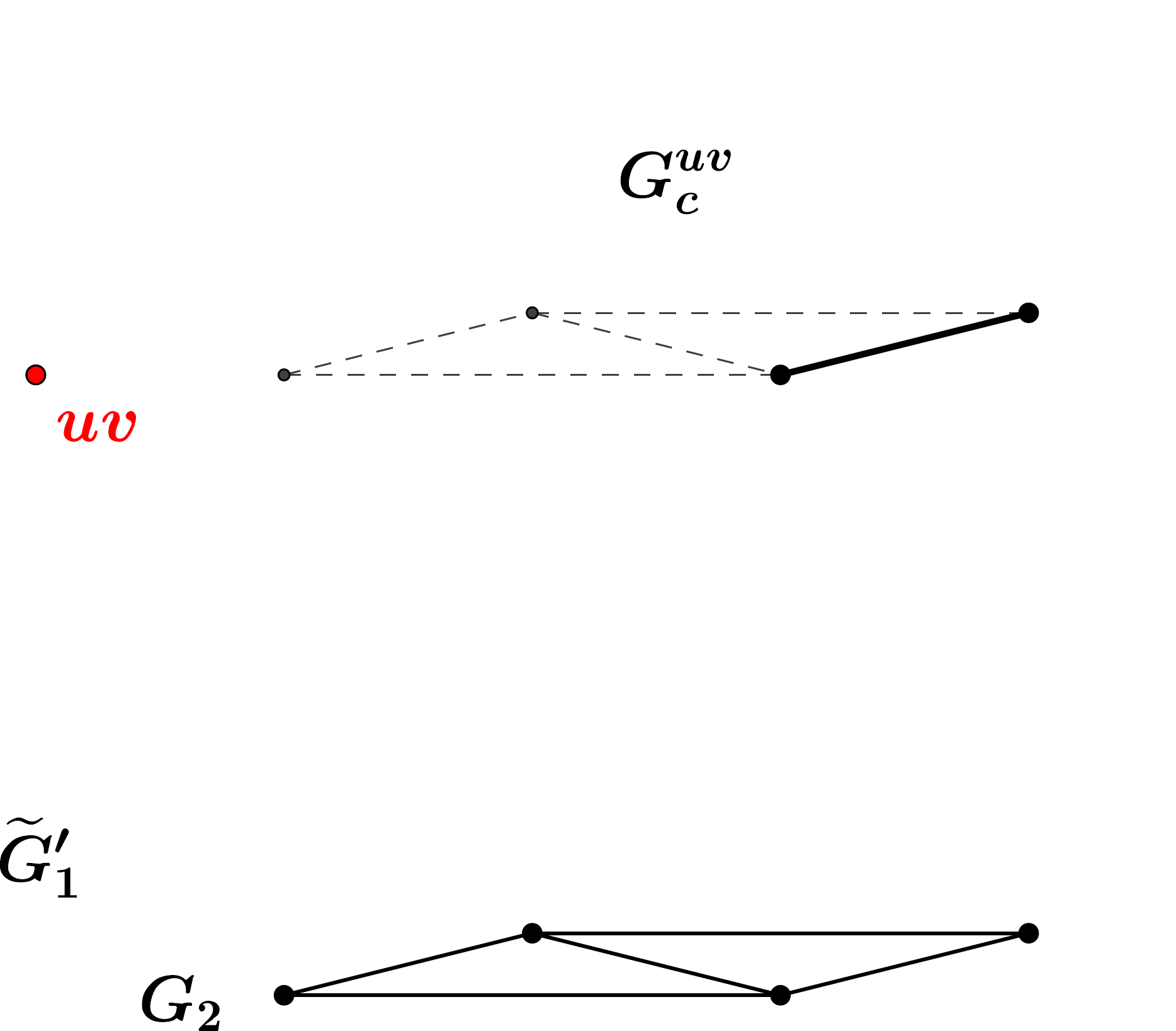}
	\end{minipage}
	
	\caption{\label{contraction}Examples of graphs $G_{uv}$ and $G^{uv}_c$}
\end{figure}

Let $G_{uv}$ be the subgraph of $\widehat
\Gamma$$:=$$G_1\product\cdots\product G_{i-1}\product
\widehat G_i\product G_{i+1}\product\cdots\product G_m$ obtained from $G$
by contracting every edge of
{\it type} $uv$, i.e., by identifying any pair of vertices of the form
$((v_1,\ldots,v_{i-1},u,v_{i+1},\ldots,v_m), 
(v_1,\ldots,v_{i-1},v,v_{i+1},\ldots,v_m))$, and removing
multiple edges. Let $G^{uv}$ be the subgraph of
$\widetilde{\Gamma}:=G_1\product\cdots\product
G_{i-1}\product \widetilde G_i\product G_{i+1}\product\cdots\product G_m$
obtained from $G$ by applying the
transformation of $G_i$ to $\widetilde G_i$. Namely, $G^{uv}$ is the subgraph of
$\widetilde{\Gamma}$ induced by the following set of vertices: \lenum{1}
$(v_1,\ldots,v_{i-1},\widetilde{w},v_{i+1},\ldots,v_m)$ is a vertex of $G^{uv}$ 
if
$(v_1,\ldots,v_{i-1},u,v_{i+1},\ldots,v_m)$ and 
$(v_1,\ldots,v_{i-1},v,v_{i+1},\ldots,v_m)$ are vertices of
$G$ and \lenum{2} $(v_1,\ldots,v_{i-1},\widetilde{x},v_{i+1},\ldots,v_m)$ is a 
vertex of $G^{uv}$ if
$(v_1,\ldots,v_{i-1},x,v_{i+1},\ldots,v_m)$ is a vertex of $G$, $x\in N$, and 
$(v_1,\ldots,u,\ldots,v_m)$
and $(v_1,\ldots,v,\ldots,v_m)$ are vertices of $G$.

Notice that $G_{uv}$ plays the role of $\mcs_e$ in the binary case. To play the 
role of $\mcs^e$ we
define the graph $G^{uv}_c$ which is the subgraph of $G^{uv}$ induced by the 
vertices that have a
central node $\widetilde{w}$ of $\widetilde G_i$ as their $i$th coordinate. If 
$G_i$ is a $K_2$ (or,
more generally, the edge $uv$ does not belong to a triangle), then $G^{uv}$ 
coincides with
$G_c^{uv}$. The remaining vertices of $G^{uv}$, those having a leaf 
$\widetilde{x}$ of $\widetilde
G_i$ as their $i$th coordinate, will be called \emph{tip vertices}. We denote 
by $V_l(G^{uv})$ the set
of \emph{tip vertices}. Finally, the remaining edges of $G^{uv}$ 
will be denoted by $E_l(G^{uv})$. Examples of graphs $G_{uv}$ and
$G^{uv}_c$ are presented in Fig. \ref{contraction}. In all subsequent results, 
we suppose that $G$ is
a subgraph of a Cartesian product $\Gamma:=G_1\product\cdots\product G_m$,
$G_i$ is any factor of
$\Gamma$, and $uv$ is any edge of $G_i$.

\begin{lemma}\label{prop_VC(G_uv)} Let $\Gamma=G_1\product\cdots\product
G_m$ be a Cartesian
	product of finite graphs and let $G$ be an induced subgraph of $\Gamma$. 
	For all $G_i$,
	$i\in\{1,\ldots,m\}$, and for all $uv\in E(G_i)$, the following 
	inequalities hold:
	\begin{enumerate}[(1)]
		\item\label{itm_vcd(G_uv)} $\vcd^*(G_{uv})\le \vcd^*(G)$;
		\item\label{itm_vcdens(G_uv)} $\vcdens^*(G_{uv})\leq\vcdens^*(G)$,
	\end{enumerate}
	\noindent where the VC-dimension and the VC-density of $G_{uv}$ are 
	computed with respect to $\widehat \Gamma$.
\end{lemma}

\begin{proof}
	We will  prove that if a minor-subproduct $M$ of $\widehat \Gamma$ is 
	shattered by $G_{uv}$,
	then $M$ is also a minor-subproduct of $\Gamma$ shattered by $G$.
	Let $M:=M_{1}\product\cdots\product M_{m}$ be defined by the partition
	$\mcp:=\mcp_{1}\times\cdots\times \mcp_{m}$ of the product 
	$\widehat\Gamma$, where
	$\mcp_{j}:=\{P_1^{j},\ldots,P_{t_j}^{j}\}$ is a partition  of $G_{j}$ (or 
	$\widehat G_i$ if $j=i$) defining
	the minor $M_{j}$, $j=1,\ldots,m$. We suppose
	that $M$ contains $k$ non-trivial factors indexed by $i_1,\ldots,i_k$.  
	Since $M$ is shattered by
	$G_{uv}$,  any set $P^{1}_{l_1}\times\cdots\times P^{m}_{l_m}$
	of $\mcp$ (where $l_1=1,\ldots,t_1,\ldots,l_m=1,\ldots,t_m$) contains a 
	vertex $x$ of $G_{uv}$.
	Recall that we merged two adjacent vertices $u$ and $v$ of the
	factor $G_i$ of $\Gamma$. We distinguish two cases.
	
	\medskip\noindent
	{\bf Case 1.} $i\notin \{i_{1},\ldots,i_{k}\}$.
	
	\noindent
	Then $M$ is also a minor-subproduct of $\Gamma$. We will prove that the 
	vertex $x$ of
	$P^{1}_{l_1}\times\cdots\times P^{m}_{l_m}$ also belongs to $G$.
	Consider the $i$th  coordinate $x_i$ of $x$. If $x_i\ne w$, then  $x$ also 
	belongs to $G$.
	Otherwise, if $x_i=w$, then at least one of the vertices 
	$(x_1,\ldots,u,\ldots,x_m)$ or
	$(x_1,\ldots,v,\ldots,x_m)$ must be a vertex of $G$, and thus we can denote 
	by $x$ that vertex.
	
	\medskip\noindent
	{\bf Case 2.} $i \in \{i_{1},\ldots,i_{k}\}$, say $i:=i_j$.
	
	\noindent
	Consider the partition $\mcp_{i}:=\{ P_1^{i},\ldots,P_{t_i}^{i}\}$ of 
	$\widehat G_{i}=\widehat G_{i_j}$
	and suppose that $w\in P_{l}^{i}$. Using the partition $\mcp_{i}$
	of $\widehat G_{i}$, we will define in the following way the partition
	$\mcp'_{i}=\{P_1^{i},\ldots,P_{l-1}^{i},P',\ldots, P_{t_i}^{i}\}$  of 
	$G_i$, where
	$P':= P_{l}^{i}\setminus \{ w\}\cup\{u,v\}$.  Let 
	$\mcp'=\mcp_{1}\times\cdots\times \mcp_{i-1}\times
	\mcp'_{i}\times \mcp_{i+1}\times \cdots\times \mcp_{m}$.
	It can be easily seen that $\mcp'$ is a partition of the Cartesian product 
	$\Gamma$, and that it
	provides a minor-subproduct representation of $M$
	(in the product $\Gamma$). Then as in Case 1 one can show that either 
	$x_i\ne w$ and $x$ is also
	a vertex of $\Gamma$ belonging to $G$, or that $x_i=w$
	and  $(x_1,\ldots,u,\ldots,x_m)$, or $(x_1,\ldots,v,\ldots,x_m)$ must
	belong to $G$.
	
	This shows that if $M$ is a minor-subproduct of $\widehat{\Gamma}$ 
	shattered by $\widehat G$, then $M$ is also a minor-subproduct of $\Gamma$ 
	shattered	by $G$, establishing the inequalities (\ref{itm_vcd(G_uv)}) 
	and (\ref{itm_vcdens(G_uv)}).
\end{proof}

\begin{lemma} \label{prop_VC(G^uv)} Let $\Gamma=G_1\product\cdots\product
G_m$ be a Cartesian
	product of finite graphs and let $G$ be an
	induced subgraph of $\Gamma$. For all $G_i$, $i\in\{1,\ldots,m\}$, and for 
	all $uv\in E(G_i)$, the
	following inequalities hold:
	\begin{enumerate}[(1)]
		\item\label{itm_vcd_Gc^uv} $\vcd^*(G^{uv}_c)\le \vcd^*(G)-1$;
		\item\label{itm_vcdens_Gc^uv}  $\vcdens^*(G^{uv}_c)\leq\vcdens^*(G)-1$,
	\end{enumerate}
	\noindent where the VC-dimension and the VC-density of $G^{uv}_c$ are 
	computed with respect to $\widetilde \Gamma$.	
\end{lemma}

\begin{proof}
	Suppose by way of contradiction that $G^{uv}_c$ shatters in  
	$\widetilde{\Gamma}$ a
	minor-product $M=M_1\product\cdots\product M_m$ such that either
	$\frac{|E(M)|}{|V(M)|}\geq\vcdens^*(G)$ or $M$ has at least $\vcd^*(G)$ 
	non-trivial factors
	$M_{i_1},\ldots,M_{i_k}$. Let  $M$ be defined by the partition
	$\mcp:=\mcp_{1}\times\cdots\times \mcp_{m}$ of $\widetilde{\Gamma}$, where
	$\mcp_{j}:=\{P_1^{j},\ldots,P_{t_j}^{j}\}$ is a partition  of $G_{j}$
	(or $\widetilde G_i $ if $j=i$) defining the minor $M_{j}$, $j=1,\ldots,m$. 
	
	Notice that all vertices of $G_c^{uv}$ have $\widetilde w$ as their $i$th 
	coordinate. Therefore if
	$i\notin\{i_1,\ldots,i_k\}$, we can define $\mcp'$ as the product of 
	partitions $\mcp$ where the
	$i$th partition $\mcp_i=V(\widetilde G_i)$ has been replaced by 
	$\mcp_i'=V(G_i)$. So $\mcp'$ is a
	partition of $\Gamma$ defining the same minor $M$ now shattered by $G$.
	Vice-versa, if $i\in\{i_1,\ldots,i_k\}$, say $i=i_k$, then, since all 
	vertices of $G_c^{uv}$ have
	$\widetilde w$ as their $i$th coordinate, the partition $\mcp_{i}$ must 
	contain a single member, i.e.,
	$M_i=K_1$. Thus further we can assume that $\mcp_i=V(\widetilde G_i),$ and
	$i\notin\{i_1,\ldots,i_k\}$.

	Now, we assert that $M':=M_{1}\product\cdots\product M_{i-1}\product
	K_2\product
	M_{i+1}\product\cdots\product M_m$ is a minor-subproduct of $\Gamma$
	shattered by the graph $G$.
	Since $M$ has $K_1$ as $i$th coordinate, this would imply that 
	$|V(M')|=2|V(M)|$ and
	$|E(M')|=2|E(M)|+|V(M)|$, whence
	\begin{gather*}
	\frac{|E(M')|}{|V(M')|} =   \frac{2|E(M)| + |V(M)|}{2|V(M)|} 
                            =   \frac{|E(M)|}{|V(M)|} + \frac{1}{2}
                            =   \vcdens^*(G_c^{uv})   + \frac{1}{2}
                            \ge \vcdens^*(G)          + \frac{1}{2},
	\end{gather*}
	leading to a contradiction with $\frac{|E(M)|}{|V(M)|}\geq\vcdens^*(G)$ and 
	showing (\ref{itm_vcdens_Gc^uv}).
	Furthermore, $M'=M\product K_2$ implies that its
	$\vcd^*(M')=\vcd^*(M)+1$ leading to a
	contradiction with $k\geq\vcd^*(G)$ and showing (\ref{itm_vcd_Gc^uv}).
	
	First we prove that $M'$ is a minor-product of $\Gamma$. Recall that
	$M=M_{1}\product\cdots\product
	M_{m}$  is defined by the partition $\mcp:=\mcp_{1}\times\cdots\times 
	\mcp_{m}$ of
	$\widetilde{\Gamma}$, where $\mcp_{i}=V(\widetilde G_i)$. Define the 
	following partition
	$\mcp':=\mcp_{1}\times\cdots\times \mcp_{i-1}\times
	\mcp'_{i}\times\mcp_{i+1}\times\cdots\times\mcp_{m}$ of  $\Gamma$, where 
	$\mcp'_{i}=\{ P_1^{i},
	P_2^{i}\}$ and $\{ P_1^{i}, P_2^{i}\}$ define a partition of $V(G_i)$ into 
	two connected subgraphs,
	the first containing the vertex $u$ and the second containing the vertex 
	$v$. This can be done by
	letting $P_1^{i}$ be the set of all vertices of $G_i$ reachable from $u$ 
	via simple paths not
	passing via $v$ and by setting $P_2^{i}:=V(G_i)\setminus P_1^{i}$.  Then 
	clearly $M'$ is defined by
	the partition $\mcp'$.

	It remains to show that $M'$ is shattered in $\Gamma$ by $G$. Pick any set
	$P^{1}_{l_1}\times\cdots\times P^{m}_{l_m}$ of $\mcp'$ (where
	$l_1=1,\ldots,t_1,\ldots,l_m=1,\ldots,t_m$ and $t_i=2$). Since $M$ is 
	shattered by $G_c^{uv}$, the
	set $P^{1}_{l_1}\times\cdots\times P^{i-1}_{l_{i-1}}\times V(\widetilde 
	G_i)\times
	P^{i+1}_{l_{i+1}}\times\cdots\times P^m_{l_m}$ has a vertex $x$ belonging 
	to $G^{uv}_c$, and then
	having $\widetilde w$ as $i$th coordinate. From the definition of 
	$G^{uv}_c$ and of the vertex
	$\widetilde w$ we conclude that the vertex $x_1$ obtained from $x$ by 
	replacing the $i$th
	coordinate $\widetilde w$ by $u$ is a vertex of 
	$P^{1}_{l_1}\times\cdots\times
	P^{i}_1\times\cdots\times P^m_{l_m}$, while the vertex $x_2$ obtained from 
	$x$ by replacing the
	$i$th coordinate $\widetilde w$ by $v$ is a vertex of 
	$P^{1}_{l_1}\times\cdots\times
	P^{i}_2\times\cdots\times P^m_{l_m}$. Since $x$ is a vertex of $G^{uv}_c$, 
	then $x_1$ and $x_2$
	are vertices of $G$. This shows that both sets 
	$P^{1}_{l_1}\times\cdots\times
	P^{i}_1\times\cdots\times P^m_{l_m}$  and $P^{1}_{l_1}\times\cdots\times 
	P^{i}_2\times\cdots\times
	P^m_{l_m}$ of $\mcp'$ contain vertices $x_1$ and $x_2$ which belong to $G$. 
	This proves that
	$G$ shatters $M'$ in $\Gamma$.
\end{proof}

\begin{lemma}\label{lem_counting_V_and_E} The graphs $G,G_{uv},G_c^{uv}$, and 
$G^{uv}$ satisfy
	the following relations:
	$$\begin{cases}
	|V(G)|=|V(G_{uv})|+|V(G^{uv})|-|V_l(G^{uv})|=|V(G_{uv})|+|V(G^{uv}_c)|, \\
	|E(G)|\leq|E(G_{uv})|+|E(G^{uv})|+|V(G^{uv}_c)|.
	\end{cases}$$
\end{lemma}

\begin{proof} $G_{uv}$ contains the vertices of $G$ minus one vertex for each 
contracted edge of
	type $uv$. $G^{uv}$ contains one vertex for each contracted edge of $G$ 
	(the vertices of
	$V(G^{uv}_c)$) plus one vertex for each triangle of $G$ involving an edge 
	of type $uv$ (the
	vertices
	of $V_l(G^{uv})$). Therefore to obtain $|V(G)|$, from  
	$|V(G_{uv})|+|V(G^{uv})|$ we have to
	subtract
	$|V_l(G^{uv})|$.
	
	$E(G_{uv})$ contains the set $E(G)$ of edges of $G$ minus \lenum{1} the 
	contracted edges of
	type $uv$, minus \lenum{2} the multiple edges obtained when contracting a 
	triangle of $G$
	containing an edge of type $uv$, minus \lenum{3} the multiple edges 
	obtained when contracting a
	square of $G$ with two opposite edges of type $uv$, and plus \lenum{4} some 
	edges we may
	have created if we had only two opposite vertices of a square, one of type 
	$u$ and the other of
	type $v$. Notice that  there are $|V(G^{uv}_c)|$ edges of type $uv$ (group 
	\lenum{1}), there are
	$|E_l(G^{uv})|$ in group \lenum{2}, and $|E(G^{uv}_c)|$ in group \lenum{3}. 
	Since
	$|E(G^{uv})|=|E_l(G^{uv})|+|E(G^{uv}_c)|$, we obtain the required inequality
	$|E(G)|\leq|E(G_{uv})|+|E(G^{uv})|+|V(G^{uv}_c)|$.
\end{proof}

\begin{lemma}\label{lem_G_c^uv} $\frac{|V(G^{uv})|}{|V(G^{uv}_c)|}\le |N|$, 
where $N$ is the set of
	the common neighbors of $u$ and $v$.
\end{lemma}

\begin{proof} Each vertex of $G^{uv}$ is either a central vertex of the form
	$(v_1,\ldots,v_{i-1},\widetilde{w},v_{i+1},\ldots,v_m)$ or is a tip vertex 
	of the form
	$(v_1,\ldots,v_{i-1},\widetilde{x},v_{i+1},\ldots,v_m).$ Each tip vertex of 
	$G^{uv}$ is adjacent to a
	single central vertex. Since each central vertex of $G^{uv}$ is adjacent
	to at most $|N|$ tip vertices, we obtain the required inequality.
\end{proof}

\subsection{Proof of Theorem \ref{thm_subgraphs_products}}

Let $G$ be a subgraph with $n$ vertices of the Cartesian product
$\Gamma:=G_1\product\cdots\product
G_m$ of connected graphs $G_1,\ldots,G_m$ from ${\mathcal G}(H)$. We have to 
prove that
$\frac{|E(G)|}{|V(G)|}\le \mu(H)\cdot\vcd^*(G)\le \mu(H)\cdot\log n$. The 
second 
inequality follows
from Lemma~\ref{VC-dim-prop2}. We will prove the inequality 
$\frac{|E(G)|}{|V(G)|}\le
\mu(H)\cdot\vcd^*(G)$ by induction on the number of vertices in the factors of 
$\Gamma$. Since each
factor $G_i$ of $\Gamma$  belongs to ${\mathcal G}(H)$, $G_i$ is 
$\mu(H)$-degenerated, i.e.,
$G_i$ and any of its subgraphs contains a vertex $v$ of degree at most
$\mu(H)$. Let $u$ be any
neighbor of $v$ in $G_i$. Then the set $N$ of common neighbors of $u$ and $v$ 
has size at most
$\mu(H)-1$. Consider the graphs  $G_{uv}, G^{uv},$ and $G^{uv}_c$ obtained from 
$G$ by
performing the  operations from previous subsection
with respect to the edge $uv$ of $G_i$. Then $G_{uv}$ is a subgraph of the 
product $\widehat
\Gamma =G_1\product\cdots\product G_{i-1}\product\widehat G_i\product
G_{i+1}\cdots\product G_m$. Since
$\widehat G_i$ is a minor of $G_i$, all factors of  $\widehat \Gamma$ belong to 
${\mathcal G}(H)$.
Moreover, since $\widehat G_i $ contains fewer vertices than $G_i$, we can
apply the induction
assumption to subgraphs of $\widehat \Gamma$, in particular to $G_{uv}$.  
Analogously, $G^{uv}$
and $G_c^{uv}$ are subgraphs of the product
$\widetilde{\Gamma}=G_1\product\cdots\product
G_{i-1}\product\widetilde G_i\product G_{i+1}\product\cdots\product G_m$
and since $\widetilde G_i$ is a star
isomorphic to a subgraph of $G_i$, all factors of $\widetilde{\Gamma}$ also 
belong to ${\mathcal
	G}(H)$. Since $\widetilde G_i$ contains fewer vertices than $G_i$, also
	the graphs $G^{uv}$ and
$G_c^{uv}$ do.  Consequently, we have $\frac{|E(G_{uv})|}{|V(G_{uv})|}\le 
\mu(H)\cdot\vcd^*(G_{uv})$
and $\frac{|E(G_c^{uv})|}{|V(G_c^{uv})|}\le \mu(H)\cdot\vcd^*(G_c^{uv})$.

By Lemma \ref{lem_counting_V_and_E} and using the inequality 
$\frac{a_1+a_2}{b_1+b_2}\le \max\{
\frac{a_1}{b_1},\frac{a_2}{b_2}\}$, we obtain
$$
\begin{aligned}
\frac{|E(G)|}{|V(G)|}	&\leq 
\frac{|E(G_{uv})|+|E(G^{uv})|+|V(G^{uv}_c)|}{|V(G_{uv})|+|V(G^{uv}_c)|}	\\
&\leq \max\left\{\frac{|E(G_{uv})|}{|V(G_{uv})|}, 
\frac{|E(G^{uv})|+|V(G^{uv}_c)|}{|V(G^{uv}_c)|} \right\}.
\end{aligned}
$$

By Lemma \ref{prop_VC(G_uv)},   $\vcd^*(G_{uv})\le \vcd^*(G)$, whence
$$
\frac{|E(G_{uv})|}{|V(G_{uv})|}\le \mu(H)\cdot\vcd^*(G_{uv})\le 
\mu(H)\cdot\vcd^*(G).
$$
Thus it remains to provide a similar upper bound for
$\frac{|E(G^{uv})|+|V(G^{uv}_c)|}{|V(G^{uv}_c)|}$. Since 
$|E(G^{uv})|+|V(G^{uv}_c)|=|E(G^{uv}_c)|+|V_l(G^{uv})|+|V(G^{uv}_c)|=|E(G^{uv}_c)|+|V(G^{uv})|$
and $|N|\le \mu(H)$, from Lemmas \ref{prop_VC(G^uv)} and  \ref{lem_G_c^uv} we 
conclude:
$$
\begin{aligned}
\frac{|E(G^{uv})| + |V(G^{uv}_c)|}{|V(G^{uv}_c)|}
&=\frac{|E(G_c^{uv})| + |V(G^{uv})|}{|V(G^{uv}_c)|}\\
&\leq \mu(H)\cdot\vcd^*(G^{uv}_c) + \mu(H)	\\
&\leq \mu(H)\cdot(\vcd^*(G)-1) + \mu(H)				\\
&= \mu(H)\cdot\vcd^*(G).
\end{aligned}
$$
This establishes the inequality  $\frac{|E(G)|}{|V(G)|}\le 
\mu(H)\cdot\vcd^*(G)$ 
and concludes the
proof of the theorem.
\qed

\begin{remark} To prove Conjecture \ref{conj_E/V<vcdens}, in Lemma 
\ref{prop_VC(G^uv)}  it is necessary to establish a stronger inequality 
(\ref{itm_vcdens_Gc^uv}).
\end{remark}

\section{Special graph classes}\label{sect_consequences}  

In this section we specify our results (and sometimes obtain stronger
formulations) for subgraphs of
Cartesian products of several classes of graphs: bounded degeneracy graphs, 
graphs with
polylogarithmic average degree, dismantlable graphs, chordal graphs, octahedra, 
cliques.

\subsection{Bounded degeneracy graphs}
The {\it degeneracy}  of a graph $G$ is the smallest $k$ such that the vertices 
of $G$ admit a total
order $v_1,\ldots,v_n$ such that the degree of $v_i$ in the subgraph of $G$ 
induced by
$v_i,\ldots,v_n$ is at most $k$ (similarly to arboricity, degeneracy of $G$ is 
linearly bounded by its
density $\dens(G)$). Let ${\mathcal G}_{\delta}$ denote the class of all graphs 
with degeneracy at
most $\delta$. ${\mathcal G}_{\delta}$ contains the class of all graphs in 
which all degrees of
vertices are bounded by $\delta$. ${\mathcal G}_{5}$ contains the class of all 
planar graphs. The
density of any graph from  ${\mathcal G}_{\delta}$ is bounded by $\delta$. From 
Theorem
\ref{thm_subgraphs_products1} we immediately obtain the following corollary:

\begin{corollary} \label{corollary_bounded-degree-graphs}
	If $G$ is a subgraph of $G_1\product\cdots\product G_m$ and
	$G_1,\ldots,G_m\in {\mathcal
		G}_{\delta}$, then $\frac{|E(G)|}{|V(G)|}\le 2\delta\cdot\log |V(G)|.$
\end{corollary}

\subsection{Graphs of polylogarithmic average degree}

Let ${\mathcal G}_{\log,k}$ denote the class of graphs $H$ closed by subgraphs 
and in which
$\mad$ is bounded by a polylogarithmic function, i.e., $2\ceil{\dens(H)}\le 
\log^k|V(H)|$. In this case,
Theorem \ref{thm_subgraphs_products1} provides the following corollary:

\begin{corollary} \label{corollary_polylog} If $G$ is a subgraph of
$G_1\product\cdots\product G_m$ and
	$G_1,\ldots,G_m\in {\mathcal G}_{\log,k}$, then $G\in {\mathcal 
	G}_{\log,k+1}$, i.e., 
	$\frac{|E(G)|}{|V(G)|}\le \log^{k+1}|V(G)|.$
\end{corollary}

\begin{proof}
	By Theorem \ref{thm_subgraphs_products1}, $\frac{|E(G)|}{|V(G)|}\leq
	\max\{\ceil{\mad(\pi_1(G))},\ldots,\ceil{\mad(\pi_m(G))}\}\cdot\log 
	|V(G)|.$ Since  $\pi_i(G)$ is a
	subgraph of $G_i$, all projections $\pi_i(G), i=1,\ldots,m,$ of $G$ on 
	factors belong to ${\mathcal
		G}_{\log,k}$. Consequently,  $\ceil{\mad(\pi_i(G))}\le 
		\log^k|V(\pi_i(G))|\le \log^k|V(G)|$. Thus
	$\frac{|E(G)|}{|V(G)|}\leq \log^{k+1}|V(G)|$.
\end{proof}

\subsection{Products of dismantlable graphs}
A vertex $u$ of a graph $G$ is {\it dominated} in $G$ by its neighbor $v$ if 
$N[u]\subseteq N[v]$
($uv$ is called a {\it dominating edge}). A graph $G$ is {\it dismantlable} 
\cite{HeNe} if $G$ admits an ordering
$v_1,v_2,\ldots,v_n$ of its vertices such that for each $i\in\{1,\ldots,n\}$, 
$v_i$ is dominated in the
subgraph $G[v_i,\ldots,v_n]$ of $G$ induced by the vertices $v_i,\ldots,v_n$. 
Dismantlable graphs
comprise chordal graphs, bridged graphs, and weakly bridged graphs as 
subclasses (for definition
of the last two classes and their dismantlability see for example 
\cite{BrChChGoOs}). Dismantlable
graphs are exactly the cop-win graphs \cite{NoWi}.

A  \emph{min-dismantling order} of $G$ is a dismantling order	in which at 
each step $i$, $v_i$ is a
dominated vertex of $G[v_i,\ldots,v_n]$ of smallest degree.   We define the 
{\it dismantling
	degeneracy} $\dd(G)$ of $G$ as the maximum degree of  a vertex $v_i$ in 
	$G[v_i,\ldots,v_n]$ in a
min-dismantling order of $G$.  Let $\Gamma:=\prod_{i=1}^m G_i$ be a Cartesian 
product of
dismantlable graphs $G_1,\ldots,G_m$. Let $v_{i,1},\ldots,v_{i,n_i}$ be a 
min-dismantling ordering
$<_i$ of the vertex-set of $G_i$, $i=1,\ldots,m$. Define a partial order 
$(\prod_{i=1}^mV_i,\preceq)$
on the vertex set  of $\Gamma$  as the Cartesian product of the totally ordered 
sets $(V_i,<_i),$
$i=1,\ldots,m$. Let $<$ denote any linear extension of $\preceq$. We will call 
this order a {\it product
	min-dismantling order}. The maximal degree of a vertex $v_i$ in the 
	subgraph of $\Gamma$
induced
by $v_i,\ldots,v_N$, where $N=n_1\times\cdots\times n_m$, will be denoted by 
$\DD(\Gamma)$. One
can show that $\DD(\Gamma)=\sum_i \dd(G_i)$.

For a subgraph $G$ of $\Gamma$ with its vertices ordered $v_1,\ldots,v_n$ 
according to the total
order $<$ on $\Gamma$, we will define $\DD(G)$ as the maximum degree of a 
vertex $v_i$ in the
subgraph  of $G$ induced by $v_i,\ldots,v_n$. Note that 
$\DD(G)\leq\DD(\Gamma)$. Note also that
contracting a dominating edge $u_iv_i$ of a factor $G_i$ gives a minor of $G_i$ 
that, at the same
time, is an induced subgraph of $G_i$. Thus, if $G$ is a subgraph of a 
Cartesian product of
dismantlable graphs and at each step in the proof of Theorem
\ref{thm_subgraphs_products} we
contract a dominating edge of a factor, we obtain the following result:

\begin{proposition}\label{prop_density_dismantlable} If $G$ is a subgraph of
	$\Gamma=G_1\product\cdots\product G_m$ and $G_1,\ldots,G_m$ are
	dismantlable graphs,  then
	$\frac{|E(G)|}{|V(G)|} \leq \DD(G)\cdot\vcd(G).$
\end{proposition}

\begin{proof}	
	Since $G^{uv}$ is a subgraph of $G$, from Lemma \ref{lem_G_c^uv}  we infer 
	that
	$\frac{|V(G^{uv})|}{|V(G^{uv}_c)|}\leq\DD(G^{uv})\leq\DD(G)$.
	Since $G_{uv}$ is an induced subgraph of $G$, we directly obtain that 
	$\vcd(G_{uv}) \leq
	\vcd(G)$. Then, according to Lemma \ref{prop_VC(G^uv)} for $\vcd(G)$,
	we have $\vcd(G^{uv}_c) \leq \vcd(G)-1$. The result now follows from those 
	remarks together with
	the previous inequality.
\end{proof}

\subsection{Products of chordal graphs}\label{sect_chordal_graph}
For {\it chordal graphs} $G$, we use the facts that any simplicial ordering of 
$G$ is a dominating
ordering and that  for any edge $uv$ of a chordal graph $G$, the graph $G_{uv}$ 
is chordal and
$\cm{G_{uv}} \leq \cm{G}$ (where $\cm G$ is the size of a largest clique of 
$G$).  As for
dismantlable graphs, contracting each time an edge between a simplicial vertex 
and its neighbor and
applying this property, we obtain the following corollary of Proposition
\ref{prop_density_dismantlable} and Theorem \ref{thm_subgraphs_products}:

\begin{corollary}\label{corollary_dens<omega.vcdens}  If $G$ is a subgraph of
	$G_1\product\cdots\product G_m$ and $G_1,\ldots,G_m$ are chordal
	graphs, then
	$\frac{|E(G)|}{|V(G)|} \leq \cm G\cdot\vcd(G)$.
\end{corollary}

One basic class of Cartesian products of chordal graphs is the class of {\it 
Hamming graphs}, i.e.,
Cartesian product of complete graphs. Consequently,  Corollary 
\ref{corollary_dens<omega.vcdens}
also applies to subgraphs of Hamming  graphs. There is a lot of literature on 
this class of graphs,
since it naturally generalize binary words of constant length (hypercubes) to 
words in arbitrary
alphabets. In particular, generalizations of shattering and ``VC-dimension" 
have been studied in the
case of Hamming graphs (see \cite{Pollard_book,Natarajan}), leading to a 
``stronger" Sauer's lemma
\cite{HaLo} and density results \cite{RuBaRu}.

\subsection{Products of octahedra}\label{sect_octahedra}

The \emph{$d$-dimensional octahedron} $K_{2,\ldots,2}$ (or the {\it cocktail 
party graph}
\cite{DeLa}) is the complete graph on $2d$ vertices minus a perfect matching. 
Equivalently,
$K_{2,\ldots,2}$ is the $d$-partite graph in which each part has two  vertices  
(which we will call {\it
	opposite}). $K_{2,\ldots,2}$ is also the 1-skeleton of the $d$-dimensional 
	cross polytope. If $e_1$
and $e_2$ are opposite vertices, it will be convenient to denote this by 
$e_2=\bar{e_1}$. A subgraph
of an octahedron is called a {\it suboctahedron}. The \emph{$d$-dimensional 
octahedron}
$K_{2,\ldots,2}$ can be viewed as the graph of an alphabet $\Sigma=\{
e_1,\bar{e_1},\ldots,e_d,\bar{e_d}\}$ with an involution 
$\varphi(e_i)=\bar{e_i}$ for each $i=1,\ldots,d$
(an {\it involution} on a finite set $X$ is a mapping $\varphi: X\rightarrow X$ 
such that
$\varphi^2(x)=x$ for any $x\in X$). Consequently, the Cartesian product of $m$ 
$d$-dimensional
octahedra can be viewed as the graph of words of length $m$ of $(\Sigma, 
\varphi)$.

In this subsection, we show that if $G$ is a subgraph of a Cartesian product 
$\Gamma$ of $m$
octahedra $G_1,\ldots,G_m$ of respective dimensions $d_1,\ldots,d_m$,  then an 
analog of Corollary
\ref{corollary_dens<omega.vcdens}  holds (for edge-isoperimetric problem in  
products of
octahedra, see \cite{Har_book}). The difference to  chordal graphs is that 
contracting an edge
of an octahedron can increase its largest clique and the result is not a 
suboctahedron anymore.
Therefore, we have to define the graphs $G_{uv}$ and $G^{uv}$ differently. If 
at the beginning all
factors are octahedra, after a few steps they will no longer be octahedra but 
suboctahedra. 
If some factor is not a clique, then it contains two opposite vertices $e$ and 
$\bar{e}$ and we can
identify them, transforming this factor into a suboctahedron with fewer
opposite pairs. If all factors
are cliques, then their product is a Hamming graph  (which can be viewed as our 
induction basis)
and we can use the results of previous subsection, namely Corollary
\ref{corollary_dens<omega.vcdens}.

Let $G_i$ be a factor of the current Cartesian product $\Gamma$ and suppose 
that $G_i$ is a
suboctahedron containing two opposite vertices $e$ and $\bar{e}$. Let $G$ be a 
subgraph of
$\Gamma$.
We will denote by $v^i$ a vertex of $\Gamma$ with all components fixed excepted 
the $i$th one, i.e.,
$v^i$ fixes the position of a copy of $G_i$. We will denote  this vertex by 
$v^i[e]$, $v^i[\bar{e}]$, or
$v^i[e']$ if we need to fix this $i$th coordinate to $e$, $\bar{e}$, or some 
neighbor  $e'\in V(G_i)$ of
$e$ and $\bar{e}$. Denote by $\widehat G_i$ the graph induced by 
$V(G_i)\setminus\{\bar{e}\}$ and
by $G_e$  the subgraph of $\widehat{\Gamma}$ in which we merged the vertices of 
$G$ having $e$
as their $i$th coordinate with those having $\bar{e}$.
Let $\widetilde G_i$ be a star with $\bar{e}$ as central vertex  and with the 
neighbors of $\bar{e}$ in
$G_i$ as \tips. Let $G^e$ be the subgraph of $\widetilde{\Gamma}$ satisfying 
the following
conditions: \lenum{1} $v^i[\bar{e}]\in V(G^e)$ if and only if 
$v^i[e],v^i[\bar{e}]\in V(G)$; \lenum{2}
$v^i[e']\in V(G^e)$ if and only if $v^i[e'], v^i[\bar{e}], v^i[e]\in V(G)$, 
i.e., we include in $G^e$ the edges of $G$ which are lost by the merging 
operation.

\begin{lemma}\label{lem_counting_V_and_E_octahedra}
	$\begin{cases}
	|V(G)|=|V(G_e)| + |V(G^e_c)| \\
	|E(G)|\leq |E(G_e)| + |E(G^e_c)| + |V_l(G^e)|.
	\end{cases}$
\end{lemma}
\begin{proof} The counting of vertices of $G$ is the same as in the proof of 
Lemma
	\ref{lem_counting_V_and_E}. The changes concern  the counting of edges. The 
	correspondence
	between the set $E(G)$ and the sets $E(G^e),E(G_e),$  and $V_l(G^e)$ is 
	illustrated in Fig.
	\ref{fig_counting_E_octahedra}. Namely, if $v^i[e],v^i[\bar{e}]\in V(G)$, 
	then the edges from these
	vertices to their neighbors  in the copy indexed by $v^i$ are counted once 
	in $E(G_e)$ and once
	in
	$V_l(G^e)$, and the edges from these vertices to other copies are counted 
	once in $E(G^e_c)$
	and
	once in $E(G_e)$. A new edge may be created if for some $v_1^i$ and $v_2^i$,
	$v_1^i[e],v_2^i[\bar{e}]\in V(G)$ but $v_1^i[\bar{e}],v_2^i[e]\notin V(G)$.
\end{proof}

\begin{figure}
	\centering
	\includegraphics[width=0.4\textwidth]{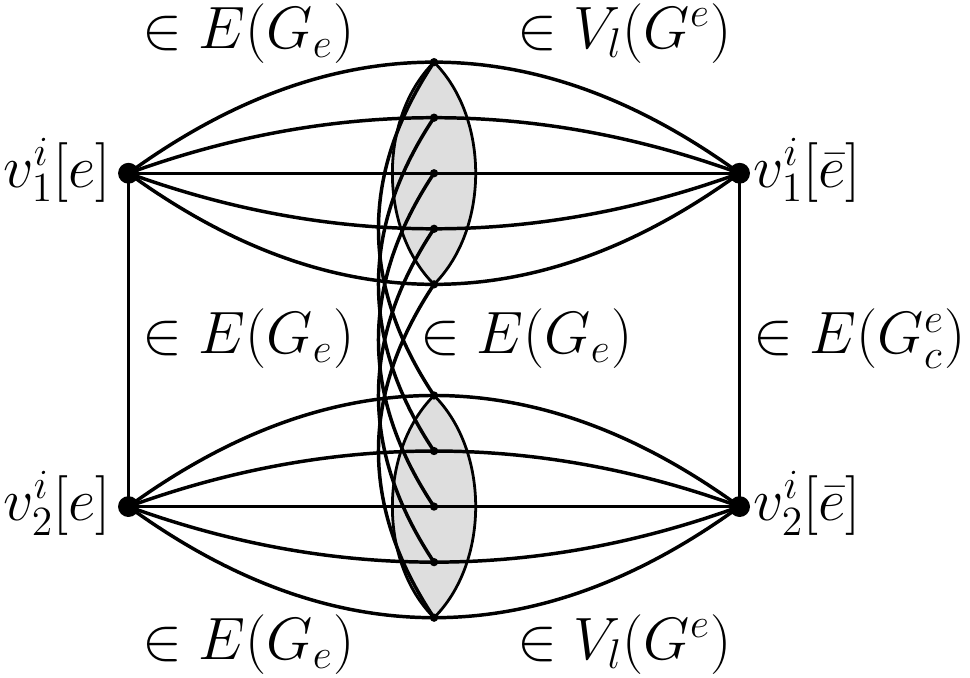}
	\caption{\label{fig_counting_E_octahedra} To the proof of Lemma
		\ref{lem_counting_V_and_E_octahedra}.}
\end{figure}

With the same arguments as in Section \ref{sect_properties_VC-dim}, we can 
prove that \lenum{1}
$\vcd(G_e)\leq \vcd(G)$ and \lenum{2} $\vcd(G^e_c)\leq\vcd(G)-1$.
Notice that $|V(G^e_c)|=0$ if and only if $G^e$ is empty. Otherwise,  if a \tip 
$v^i[e']$ exists in
$G^e$, the central node $v^i[\bar{e}]$ also exists by definition (second 
condition). Those \tips are
the neighbors (in the copy of $\widetilde G_i$ indexed by $v^i$) of 
$v^i[\bar{e}]$ in $G^e$ and
clearly they can not be more than $\cm G$, showing that 
$\frac{|V_l(G^e)|}{|V(G^e_c)|} \leq \cm G$.
Finally, using Lemma \ref{lem_counting_V_and_E_octahedra}, the inequalities (1) 
and (2), and this last
inequality, we obtain:

\begin{proposition}\label{thm_octahedra}
	$\frac{|E(G)|}{|V(G)|} \leq \cm{G}\cdot\vcd(G)$.
\end{proposition}

\section{Arboricity and adjacency labeling schemes} 
\label{sect_adjacency_scheme}

An \emph{adjacency labeling scheme} on a graph family $\mcg$ consists of a 
\emph{coding
	function} $C_G:V(G)\to\{0,1\}^*$ from the set $V(G)$ of vertices of $G$ the 
	set $\{0,1\}^*$ of finite binary words
that gives to every vertex of a graph $G$ of $\mcg$ a label, and
a
\emph{decoding function} $D_G:\{0,1\}^*\times\{0,1\}^*\to\{0,1\}$ that, given 
the labels of two
vertices of $G$, can determine whether they encode adjacent vertices or not. If 
$\mcg$ is the family
of all forests on $n$ vertices, it is easy to build an adjacency labeling 
scheme using labels of size
$2\lceil\log n\rceil$ bits. Indeed, to construct such a scheme, the coding 
function gives to every
vertex a unique id (that requires $\lceil\log n\rceil$ bits) and concatenate to 
the label of each vertex
the label of its parent. Given two labels, the decoding function determines if 
they encode adjacent
vertices by testing the equality between the first half of one label and the 
second half of the other. It
has been shown in \cite{AlDaKn} that the family of forests admits an adjacency 
labeling scheme
using unique labels of size $\lceil\log n\rceil + O(1)$ bits.  

Kannan, Naor and Rudich \cite{KaNaRu} noticed that if a graph $G$ is covered by 
$k$ forests, then
one can build an adjacency labeling scheme with $(k+1)\lceil\log n\rceil$ bits 
by applying the
construction mentioned above to each of the forests.
The \emph{arboricity} $\arbo(G)$ of a graph $G$ is the minimal number of 
forests necessary to
cover the edges of $G$. The classical theorem by Nash-William \cite{NaWi} 
asserts that the
arboricity of a graph is almost equivalent to its density:
\begin{theorem}[Nash-Williams]\label{thm_NW} The edges of a graph 
$G=(V(G),E(G))$ can be
	partitioned in $k$ forests if and only if $|E(G')|\leq k(|V(G')|-1)$ for 
	all $G'$ subgraph of $G$. i.e.,
	$
	\arbo(G) = \max_{G'\subseteq G}\left(\frac{|E(G')|}{|V(G')|-1}\right)
	\leq \dens(G).
	$
\end{theorem}

Thus, the upper bounds  for the densities of  graphs families provided in 
previous sections also
bound their arboricity. We then directly obtain upper bounds on the size of 
labels of adjacency
labeling schemes as a corollary of our results and the results of Nash-Williams 
\cite{NaWi} and Kannan and al. \cite{KaNaRu}:

\begin{corollary}
	Let $\Gamma:=G_1\product \cdots\product G_m$ and let $G$ be a subgraph
	of $\Gamma$ with $n$ vertices and induced (resp., minor) VC-dimension $d$ 
	(resp., $d^*$). Then, $G$ admits an adjacency labeling scheme with labels of
	size:
	\begin{enumerate}
		\itemsep0ex
		\item $O(d\cdot\log n)$, if $\Gamma$ is an hypercube.
		\item $O(\beta_0\cdot d\cdot\log n)$, where $\beta_0 =
		\lceil\max\{\mad(\pi_1(G)),\ldots,\mad(\pi_m(G))\}\rceil$.
		\item $O(d^*\cdot\mu(H)\cdot\log n)$, if $G_1,\ldots,G_m\in {\mathcal 
		G}(H)$.
		\item $O(\delta\cdot\log^2n)$, if $G_1,\ldots,G_m\in\mcg_\delta$.
		\item $O(\log^{k+2}n)$, if $G_1,\ldots,G_m\in\mcg_{\log,k}$.
		\item $O(\DD(G)\cdot d\cdot \log n)$, if $G_1,\ldots,G_m$ are 
		dismantlable.
		\item $O(\cm G\cdot d\cdot\log n)$, if $G_1,\ldots,G_m$ are chordal
		graphs.
        \item $O(\cm G\cdot d\cdot\log n)$, if $G_1,\ldots,G_m$ are
        octahedra. 
	\end{enumerate}
\end{corollary}

The assertions (1), (2), (3), (6), and (7) follow from Theorem 
\ref{thm_subgraphs_products}, the assertions (4) and (5) follow from 
Theorem \ref{thm_subgraphs_products1}, and the assertion (8) is a 
consequence of Proposition \ref{thm_octahedra}.

\medskip\noindent
{\bf Acknowledgements.} We would like to acknowledge the 
two referees for a careful reading of the first version and useful 
comments. This work was supported by ANR project DISTANCIA 
(ANR-17-CE40-0015). 

\bibliographystyle{amsalpha}

\end{document}